\documentclass[oribibl,a4paper,envcountsame]{llncs}

\makeatletter
\@twosidefalse
\@mparswitchfalse
\makeatother

\newif\ifcomments
\commentsfalse

\newif\iffinal
\finalfalse

\title{On homomorphic encryption using abelian groups: \newline
Classical security analysis}

\author{Eleni Agathocleous$^1$, Vishnupriya Anupindi$^2$, Annette Bachmayr$^3$, Chloe Martindale$^4$, Rahinatou Yuh Njah Nchiwo$^5$, and Mima Stanojkovski$^6$}

\institute{$^1$CISPA Helmholtz Center for Information Security, $^2$RICAM Austrian Academy of Sciences, $^3$RWTH Aachen University, $^4$University of Bristol, $^5$Aalto University, $^6$Universit\`a di Trento 
}

\usepackage[dvipsnames]{xcolor}
\definecolor{linkcolor}{rgb}{0.65,0,0}
\definecolor{citecolor}{rgb}{0,0.65,0}
\definecolor{urlcolor}{rgb}{0,0,0.65}
\usepackage[
        pdftitle={TODO},
    colorlinks=true, linkcolor=linkcolor, urlcolor=urlcolor, citecolor=citecolor
]{hyperref}
\usepackage[hyphenbreaks]{breakurl} 

\definecolor{myyellow}{rgb}{1.0, 0.75, 0.0}
\definecolor{mygreen}{rgb}{0.35, 0.71, 0.0}

\usepackage[utf8]{inputenc}
\usepackage[T1]{fontenc}
\usepackage{amsmath,mathtools}
\usepackage{amssymb}
\usepackage{mathrsfs}
\usepackage{hyperref}
\usepackage[style=english]{csquotes}
\usepackage{graphicx}
\usepackage{enumerate}
\usepackage{lipsum}
\usepackage[british]{babel}
\addto{\captionsbritish}{}
\usepackage{tipa}
\usepackage{mathtools}
\usepackage{amssymb}
\usepackage{xcolor}
\usepackage{nth}
\usepackage{tikz}
\usepackage{tikz-cd}
\usepackage{comment}
\usepackage[open,openlevel=1]{bookmark}
\usepackage{multirow}
\usepackage{appendix}
\usepackage{listings}
\usepackage{setspace}
\usepackage{pifont}
\usepackage{bm}
\usepackage{adjustbox}
\usepackage{tkz-euclide}

\usepackage{booktabs}
\usepackage{multirow}

\usepackage{enumitem}
\setitemize  {topsep=1ex,itemsep=.5ex}
\setenumerate{topsep=1ex,itemsep=.5ex}

\usepackage[capitalise]{cleveref}

\usepackage[ruled,vlined,linesnumbered]{algorithm2e}
\usepackage{algpseudocode}
\SetEndCharOfAlgoLine{}
\DecMargin{2.2mm}
\SetKwInput{Input}{Input}\SetKwInput{Output}{Output}
\SetKwFor{While}{While}{do}{}
\SetKwFor{ForEach}{For each}{do}{}
\SetKwIF{If}{ElseIf}{Else}{If}{then}{else if}{else}{endif}
\SetKw{Return}{Return}
\SetKw{KwRet}{return}
\SetKw{Recurse}{Recurse}
\SetKwFor{For}{For}{do}{}
\SetKw{Break}{break}

\newlength\InputNewLineIndent

\usepackage{geometry}

\usepackage{relsize}

\spnewtheorem{heuristic}{Heuristic}{\normalfont\bfseries}{\itshape}
\undef\problem\spnewtheorem{problem}{Problem}{\normalfont\bfseries}{\normalfont}

\usepackage{listings}
\usepackage{xspace}

\newcommand\Z{\mathbb Z}

\newcommand\F{\mathbb F}

\newcommand{\LHNPKE}{\textsf{LHN-PKE}}

\DeclareMathOperator{\Hom}{Hom}

\DeclareMathOperator{\GL}{GL}

\renewcommand\phi\varphi    

\newcommand\inabox[1]{%
    \ifvmode\addvspace{2ex}\else\vspace{2ex}\fi%
    \fcolorbox{black}{white}{%
        \begin{minipage}{0.98\textwidth}%
            {#1}%
        \end{minipage}%
    }%
    \par\addvspace{2ex}
}

\ifcomments
\newcommand{\Chloe}[1]{\textcolor{orange}{{\sf (Chloe:} {{#1})}}}
\newcommand{\Vishnu}[1]{\textcolor{Periwinkle}{{\sf (Vishnu:} {{#1})}}}
\newcommand{\Rahi}[1]{\textcolor{blue}{{\sf (Rahi:} {{#1})}}}
\newcommand{\Eleni}[1]{\textcolor{red}{{\sf (Eleni:} {{#1})}}}
\newcommand{\Mima}[1]{\textcolor{LimeGreen}{{\sf (Mima:} {{#1})}}}
\newcommand{\Annette}[1]{\textcolor{purple}{{\sf (Annette:} {{#1})}}}
\else
\newcommand{\Chloe}[1]{}
\newcommand{\Vishnu}[1]{}
\newcommand{\Rahi}[1]{}
\newcommand{\Eleni}[1]{}
\newcommand{\Mima}[1]{}
\newcommand{\Annette}[1]{}
\fi

\newtheorem{rmk}[remark]{Remark}

\newcommand{\erase}[1]{}
\newcommand{\gen}[1]{\langle #1\rangle}

\makeatletter
\def\namedlabel#1#2{\begingroup
   \def\@currentlabel{#2}%
   \label{#1}\endgroup
}
\makeatother

\begin{document}

\maketitle

\begingroup
\makeatletter
\def\@thefnmark{$\ast$}\relax
\@footnotetext{\relax
Author list in alphabetical order; see
\url{https://www.ams.org/profession/leaders/culture/CultureStatement04.pdf}.
\def\ymdtoday{\leavevmode\hbox{\the\year-\twodigits\month-\twodigits\day}}\def\twodigits#1{\ifnum#1<10 0\fi\the#1}%
\newline
Date of this document: \today.
}
\endgroup

\begin{abstract}
Leonardi and Ruiz-Lopez recently proposed an additively homomorphic public key encryption scheme based on combining group homomomorphisms with noise. \Mima{check new starting sentence and erase comments if approved}
% In \cite{Eprint}, Leonardi and Ruiz-Lopez propose an additively homomorphic public key encryption scheme whose security is expected to depend on the hardness of the \emph{learning homomorphism with noise problem} (LHN)\Rahi{I think we may need to rephrase this since it's not based on LHN}\Mima{Maybe we can say 'In \cite{Eprint}, Leonardi and Ruiz-Lopez propose an additively homomorphic public key encryption scheme based on combining group homomomorphisms with noise.'}. \Annette{I think this is a good suggestion. And if we do change the abstract we could also just leave out the [20] and say 'Leonardi and Ruiz-Lopez recently proposed...'}
Choosing parameters for their primitive requires choosing three groups $G$, $H$, and $K$.  In their paper, Leonardi and Ruiz-Lopez claim that, when $G$, $H$, and $K$ are abelian, 
then their public-key cryptosystem is not quantum secure. 
In this paper, we study security for finite abelian groups $G$, $H$, and $K$ in the classical case. 
Moreover, we study quantum attacks on instantiations with solvable~groups.
\end{abstract}

\setcounter{footnote}{0}

\subsubsection*{Acknowledgements.} We warmly thank the organizers of \emph{Women in Numbers Europe 4} for putting together this team. We also wish to thank Chris Leonardi and Andrew Sutherland for helpful conversations around the contents of \cite{Eprint} and \cite{Sutherland} respectively. In addition, we thank the three anonymous referees for their comments, which helped improve the exposition of this paper. 

The first author was supported by the European Union’s H2020 Programme under grant agreement number ERC-669891.
The second author was supported by the Austrian Science Fund, Project P34808.
The fifth author was supported by the Magnus Ehrnrooth grant 336005, the Academy of Finland project grant (351271, PI Camilla Hollanti) and by the MATINE project grant, Ministry of Defence Finland (2500M-0147, PI Camilla Hollanti).
The last author was partially supported by the Deutsche Forschungsgemeinschaft (DFG, German Research Foundation) -- Project-ID 286237555 -- TRR 195 and by the Italian program Rita Levi Montalcini for young researchers, Edition 2020.

\section*{Introduction}

Homomorphic encryption is a method of encrypting plaintext that allows users to compute directly with the ciphertext.
This has many interesting applications, including being able to engage in cloud computing without giving up your data to the owner of the cloud.
Scientifically, the premise is easy to describe:
Suppose that the plaintext and the ciphertext space both have a ring structure, and that we encrypt plaintext via a map between these spaces.
If this map is a ring homomorphism, then this describes a \emph{fully homomorphic encryption} scheme.
Creating such a ring homomorphism that describes secure encryption (requiring, for example, that such a map should be efficiently computable and hard to invert) is, however, much harder than describing its properties.
\Rahi{ In particular, jumping from only additive (or only multiplicative) to additive and multiplicative properties (enabling FHE) is a very difficult task. Despite the fact that additive (or only multiplicative) solutions existed quite early (cite it here), FHE has been an open problem for a long time.} \Mima{OK}
The closest the scientific community has come to constructing an example of fully homomorphic encryption is using maps based on (variants of) the \emph{Learning With Errors} (LWE) problem from lattice-based cryptography \cite{Regev, LPR, GSW, FHEW, TFHE}. 
However, although practical fully homomorphic encryption can be achieved from LWE-based constructions~\cite{practicalFHE}, all known such constructions are not naturally fully homomorphic:\Rahi{ mostly due to the multiplicative property that leads to a big noise growth} \Mima{OK}\Chloe{ I would prefer not to add this here, it is already later in the sentence in a more precise way. The multiplication only leads to big noise growth because it is lots of additions.}\Rahi{: Ok by me}
Decrypting a message that was encrypted using LWE relies on the `error' that was used in the encryption being small, and adding and especially multiplying encrypted messages together causes the error to grow. 
Once the error is too large, the data can no longer be decrypted, so methods such as bootstrapping need to be employed to correct this growth (see e.g.~\cite{FHEW}).
These methods, although somewhat ad hoc, may still be the best solution for providing practical fully homomorphic encryption in the near future.

In this paper, we explore an alternative approach for homomorphic encryption, 
introduced by Leonardi and Ruiz-Lopez in~\cite{Eprint}. The interested reader may find more on homomorphic encryption in \cite{FonGal}. %\Mima{I suggest removing the thank you note, we already thank them in the letter.}
The construction of Leonardi and Ruiz-Lopez, relies on something similar to the \emph{Learning Homomorphisms with Noise} (LHN) problem
introduced by Baumslag, Fazio, Nicolosi, Shpilrain, and Skeith in~\cite{Baumslag}:
Roughly speaking, this is the problem of recovering a group homomorphism from the knowledge of the images of certain elements multiplied by noise.
The focus of~\cite{Eprint} is on the difficulty of constructing post-quantum secure instantiations of their primitive,
but we believe the construction is interesting even in a classical setting.
A big advantage of the Leonardi--Ruiz-Lopez approach over a LWE approach is that the noise, 
which plays the role of the errors in LWE-based homomorphic encryption, does not grow with repeated computation. 
As such, there is no limitation on the number of additions that can be computed on encrypted data.
However, it is not clear if this construction can be extended to multiplicative homomorphic encryption. 
In general extending additive homomorphic encryption to fully homomomorphic encryption is a hard problem: LWE-based homomorphic encryption is the only known construction. 
As such, the Leonardi--Ruiz-Lopez approach is akin in some sense to the Benaloh~\cite{Benaloh} or Paillier~\cite{Paillier} cryptosystems.
In the nonabelian setting, Leonardi and Ruiz-Lopez's construction has some hopes of being post-quantum secure unlike the Benaloh or Paillier constructions, 
but as they explored already in their work this is nontrivial to instantiate, 
and our work only strengthens this claim as we show that even solvable groups may admit quantum~attacks.

$\,$\\
\noindent
Our main contributions address \textbf{finite groups} and include:
\begin{enumerate}
    \item Reducing the security of an instantiation of Leonardi and Ruiz-Lopez's public key homomorphic encryption scheme in which one of the groups is abelian (in particular, the group in which the noise lives) to the discrete logarithm problem in 2-groups (under certain plausible assumptions);
    this gives a polynomial-time classical attack if the 2-part of the relevant group is cyclic, and a practical classical attack if it is a product of a small number of cyclic groups.
    See Theorem~\ref{main-theorem} and Section~\ref{sec:limitations}.
    \item Highlighting an abelian group instantiation of Leonardi and Ruiz-Lopez's homomorphic encryption scheme where there is no known practical classical attack, 
    namely, the product of many cyclic 2-groups. This may be of interest to the community as a new example of unbounded additively homomorphic encryption.
    \item Highlighting assumptions that need to be made in order to apply any discrete-logarithm derived attack, with a view to constructing (more) examples of groups on which there is no known practical attack on Leonardi and Ruiz-Lopez's homomorphic encryption scheme.
    \item A description of a quantum attack on an instantiation with  solvable groups, under certain assumptions.
\end{enumerate}
Note that all our contributions are studying the hardness of the $\LHNPKE$ problem, given in Definition~\ref{def-LHN-PKE}, not the general LHN problem, as this is the relevant underlying ``hard problem" in Leonardi--Ruiz-Lopez encryption.

\Chloe{Do we add anything to the state-of-the-art about the noiseless case in the end? Does the noise have any relevance in a pre-quantum scenario? We should add a remark about this}

\Mima{'known generators' is not the same as (A3) or, better, if we want to use it like that we should define it. I find it very imprecise as it is.}

\Mima{I don't get the cyclic part in the table. Wasn't the whole point that we reduce it to the $2$-part and that makes things easier?}

\Mima{The solvable case 'suggests' that solvable groups are not good but we do not have a proof for this, as far as I am aware. }

\Chloe{Table needs updating now that \cref{main-theorem} has been polished a bit. Or maybe removed, if the Theorem has made things clearer?} \Mima{I am in favor of removing the table - it was a nice idea but it is too difficult to be precise at this stage of the paper}
\Chloe{removed}

The layout of this paper is as follows: In Section~\ref{sec:cryptosystem}, we recap the public key homomorphic encryption scheme proposed by Leonardi and Ruiz-Lopez in~\cite{Eprint}. 
In Section~\ref{sec:2}, we discuss some simple instantiations: The abelian case and the noiseless case. In this section we also give some basic security requirements (including some recalled from~\cite{Eprint}).
In Section~\ref{abelian-case} we describe our reduction from the abelian group instantiation of the Leonardi Ruiz-Lopez primitive to the (extended) discrete logarithm problem, under certain assumptions. 
In \cref{sectionConvert}, we describe some ways of instantiating the primitive with nonabelian groups to which our attack on abelian groups would also apply.
In Section~\ref{sec:solvable}, we describe our quantum attack on instantiations with solvable groups, under certain assumptions on how such groups would be represented.
In Section~\ref{future-work}, we outline our plans for future work.

\section{Leonardi--Ruiz-Lopez encryption}\label{sec:cryptosystem}

In this section, we describe the public key additive homomorphic encryption of Leonardi and Ruiz-Lopez~\cite[Sec.\ 5.2]{Eprint};
we will refer to this throughout this work as \emph{Leonardi--Ruiz-Lopez encryption}.
Note that this paper shows that the following encryption scheme is \emph{not} secure (even classically) in most natural abelian instantiations, see \cref{main-theorem},
and we make no claims on the security of the nonabelian case except for some warnings, see \cref{sec:limitations}, \cref{sectionConvert}, and \cref{sec:solvable}.
In the course of this work we also discovered the necessity for some extra points in the basic Leonardi--Ruiz-Lopez encryption in order to avoid security problems even in the most general case; 
our additions are marked with $\,^{\textup{new}}$ (both when they are completely new and when they are clarifications).
These points are discussed in later subsections.

Fix three finitely generated groups $G, H, K$ and probability distributions $\xi$ on $G$ and $\chi$ on $H$. This data should be chosen in such a way that operations can be performed efficiently in the groups and we can sample from both distributions efficiently. A natural choice could be, for instance, to take $G,H,K$ finite and $\xi,\chi$ to be uniform distributions. 
The groups $G,H,K$ and the distributions $\xi,\chi$ are public. 
These groups must also satisfy security assumptions \ref{it:S1}$^{\text{new}}$ and \ref{it:S3} \Mima{does this also get a $\,^{\text{new}}$?}\Chloe{It didn't have a $\,^{\text{new}}$ because it is already in \cite{Eprint}, albeit not super consistently} discussed below.
In the following sections we will mostly work with finite groups and we will always make it clear when this is the case. 

\medskip

\noindent
For the \textbf{key generation}, Alice 
\begin{itemize}
    \item chooses efficiently computable secret homomorphisms $\varphi \colon G \to H$ and $\psi \colon H \to K$  such that she can efficiently sample from $\ker(\psi)$ 
    and such that the center $\mathrm{Z}(H)$ of $H$ is not contained in $\ker(\psi)$; 
    \item chooses a positive integer $m$;
    \item samples elements $g_1,\dots, g_m \in G$ via $\xi$ and secret  elements $h_1,\dots,h_m \in \ker(\psi)$ via $\chi$;
    \item chooses an element $\tau \in \mathrm{Z}(H)\backslash \ker(\psi)$ of order $2$\footnote{The element $\tau$ having order $2$ ensures that the encryption is additive as explained in \cref{sec:homomorphic}; the construction is also viable without this requirement.} \Mima{I would rephrase the footnote as: 'The element $\tau$ having order $2$ ensures that the encryption is additive as explained in \cite[Section~5.3]{Eprint}; the construction is also viable without this requirement.}\Chloe{done}
        that satisfies \ref{it:S4}$^\text{new}$ described below.
\end{itemize}
Alice computes the public key as the set

\[ \{(g_1,\varphi(g_1)h_1), \dots, (g_m,\varphi(g_m)h_m), \tau \}.    \]
Note that, whereas the elements $g_1,\dots,g_m$ are public,  both $\varphi$ and $h_1,\dots,h_m$ are private (as are also $\psi$ and $\ker(\psi)$).

\medskip

\noindent
For \textbf{encrypting} a one-bit message $\beta \in \{0,1\}$, Bob chooses a natural number $\ell$, then samples a word $w=w_1\cdots w_{\ell}$\footnote{The sample space and method must satisfy security assumption \ref{it:S2}$^\text{new}$ described below.} over the indices $ \{1,\dots, m\}$ and using Alice's public key, he computes
\[ (g,h')=(g_{w_1}\cdots g_{w_{\ell}},\varphi(g_{w_1})h_{w_1}\cdots \varphi(g_{w_{\ell}})h_{w_{\ell}} ).
\] He then sends $(g,h)=(g,h'\tau^\beta)$ to Alice. 

\medskip

\noindent
For \textbf{decrypting} $(g,h)$, Alice computes $\nu=\psi(\varphi(g))^{-1}\cdot\psi(h) \in K$ and deduces that the message $\beta$ equals $0$ if $\nu$ equals $1_K$ (and else, $\beta$ equals $1$). 

\medskip

To see that the decryption indeed produces the correct message $\beta$, recall that $h_1,\dots, h_m$ are contained in $\ker(\psi)$.  Hence $\nu=\psi(\tau)^\beta$  and, since $\tau$ is not contained in $\ker(\psi)$, the element $\nu$  equals $1$ if and only if $\beta$ equals $0$.

For the convenience of the reader, we give a schematic summary of the data described above:
\begin{center}
\renewcommand{\arraystretch}{1.5}%
\begin{tabular}{ |c|c|c| } 
 \hline 
\ \  Public information \ \  & \ \  Alice's private information \ \  & \ \  Bob's private information \ \ \\
\hline 
$G, H, K, \xi, \chi$ & $\phi:G\rightarrow H$  &  \\ 
  \ $(g_1,\varphi(g_1)h_1), \dots, (g_m,\varphi(g_m)h_m)$  \  & $\psi:H\rightarrow K$  & $\beta\in\{0,1\}$ \\ 
$\tau\in \mathrm{Z}(H)\setminus\ker(\psi)$ &  $\ker(\psi)$&  word $w$\\
$(g,h)=(g,h'\tau^\beta) \in G \times H$ & $h_1,\ldots,h_m\in\ker(\psi)$ & \\
 \hline
\end{tabular}
\end{center}

\subsection{Homomorphic properties}\label{sec:homomorphic}

The primary selling point of Leonardi--Ruiz-Lopez encryption is that it is \emph{unbounded additive homomorphic}.
We say that an encryption function $E$ from plaintext to ciphertext space is \emph{additive} if the plaintext space admits an additive operator $+$ and, given encryptions $E(\beta)$ and $E(\tilde{\beta})$ of messages $\beta$ and $\tilde{\beta}$ respectively, one can compute a valid encryption $E(\beta + \tilde{\beta})$ of $\beta + \tilde{\beta}$ (without the knowledge of the plaintext $\beta + \tilde{\beta}$).
We say that $E$ is \emph{unbounded additive homomorphic} if such additions can be performed an unbounded number 
of times without introducing systematic decryption failures.\footnote{
The number of additions is bounded in, for example, LWE-based homomorphic encryption, where the error grows too large.
}

%The reader may have observed that
As seen above, $\tau$ having order two is not necessary for successful decryption; 
this property is needed to make the encryption additive, 
as we now recall from~\cite{Eprint}.

Write the encryptions of $\beta$ and $\tilde{\beta}$ sampled from $\{0,1\}$ as 
\[(g,h'\tau^{\beta}) = (g_{w_1}\cdots g_{w_{\ell}},\varphi(g_{w_1})h_{w_1}\cdots \varphi(g_{w_{\ell}})h_{w_{\ell}} \tau^\beta)\]
and
\[(\tilde{g},\tilde{h}'\tau^{\tilde{\beta}}) = ({g}_{\tilde{w}_1}\cdots {g}_{\tilde{w}_{\tilde{\ell}}}, \varphi({g}_{\tilde{w}_1}){h}_{\tilde{w}_1}\cdots \varphi({g}_{\tilde{w}_{\tilde{\ell}}}){h}_{\tilde{w}_{\tilde{\ell}}}  \tau^{\tilde{\beta}})\]
respectively.
Then, as $\tau$ is central in $H$ and has order 2,
we can construct a valid encryption of $\beta + \tilde{\beta}$ via the observation that
\begin{align*}
    &\varphi(g_{w_1})h_{w_1}\cdots \varphi(g_{w_{\ell}})h_{w_{\ell}} \tau^\beta
{\varphi}({g}_{\tilde{w}_1})h_{\tilde{w}_1}\cdots {\varphi}({g}_{\tilde{w}_{\tilde{\ell}}}){h}_{\tilde{w}_{\tilde{\ell}}}  \tau^{\tilde{\beta}}
\\=
&\,\varphi(g_{w_1})h_{w_1}\cdots \varphi(g_{w_{\ell}})h_{w_{\ell}}
{\varphi}({g}_{\tilde{w}_1}){h}_{\tilde{w}_1}\cdots {\varphi}({g}_{\tilde{w}_{\tilde{\ell}}}){h}_{\tilde{w}_{\tilde{\ell}}} 
\tau^{\beta + \tilde{\beta}};
\end{align*}
this encryption is given by
\[(g,h'\tau^\beta)(\tilde{g},\tilde{h}'\tau^{\tilde{\beta}})=(g\tilde{g},h'\tau^{\beta}\tilde{h'}\tau^{\tilde{\beta}})=(g\tilde{g},h'\tilde{h'}\tau^{\beta+\tilde{\beta}}).\]

\begin{rmk}
    Note that, for Leonardi--Ruiz-Lopez encryption to be fully homomorphic, it would also need to be \emph{multiplicative}: That is, at the very least, 
    given valid encryptions of $E(\beta)$ and $E(\tilde{\beta})$, we should be able to deduce a valid encryption of $\beta\tilde{\beta}$. Since we only have one group operation on the cipher text space $G\times H$, we cannot expect to find a general strategy that provides fully homomorphic encryption for arbitrary abstract groups $G$, $H$. 
    Our encryption function from plaintext to ciphertext space is
    \[E: \{0,1\} \rightarrow G\times H,\]
    where the domain can be naturally endowed with a ring structure using addition and multiplication mod 2, 
    but there is no natural 
  way of adding extra data on
     $G\times H$ that would allow us to deduce a valid encryption of $\beta\tilde{\beta}$. We stress that in the case of LWE-based homomorphic encryption, both plaintext and ciphertext spaces come equipped with a ring structure, so the equivalent of our function $E$ is generally taken to be a ring homomorphism. 
    
    The map $E$ being a ring homomorphism is, however, not always strictly necessary to deduce a valid encryption of $\beta\tilde{\beta}$.
    If in future work we were to succeed in deducing a valid encryption of $\beta \tilde{\beta}$, we expect that this will only apply to a specific instantiation of
    Leonardi--Ruiz-Lopez encryption, not one for abstract groups, where the groups provide more structure that can be exploited. 
\end{rmk}

\subsection{Remarks on Leonardi--Ruiz-Lopez encryption}

\noindent
Some remarks on the construction above: 
\begin{enumerate}[label=$(\mathrm{R}\arabic*)$]
    \item The underlying hard problem of this encryption scheme is described as the $\LHNPKE$ problem, so named as it is based on the Learning Homomorphisms with Noise problem (LHN) but is adapted to this Public Key Encryption scheme (PKE).

\begin{definition}\label{def-LHN-PKE}
    Let $G, H, K, \xi,$ and $ \chi$ be as above. We define the \emph{$\LHNPKE$} problem for $G, H, K, \xi,$ and $ \chi$ to be: Given
    $G, H, K, \xi,$ and $\chi$, 
    for any
    \begin{itemize}
        \item $\varphi$ sampled uniformly at random from $\Hom(G,H)$, 
        \item $\psi$ sampled uniformly at random from the elements of $\Hom(H,K)$ whose kernel does not contain $\mathrm{Z}(H)$,
        \item $g_1,\ldots,g_m$ sampled from $G$ using $\xi$, 
        \item $h_1,\ldots,h_m$ sampled from $\ker(\psi)$ using $\chi$, 
        \item $\tau$ sampled from the order $2$ elements of $\mathrm{Z}(H) \setminus \ker(\psi)$ using $\chi$, 
        \item $\beta$ sampled uniformly at random from $\{0,1\}$, 
        \item small $\ell$ and word $w=w_1\cdots w_{\ell}$ sampled uniformly at random from $\{1,\cdots,m\}^{\ell}$,
    \end{itemize}
    recover $\beta$ from the following information:
    \begin{itemize}
        \item $(g_1,\varphi(g_1)h_1), \ldots, (g_m,\varphi(g_m)h_m)$, 
        \item $\tau$, 
        \item $g=g_{w_1}\cdots g_{w_{\ell}}$, 
        \item $h = h'\tau^\beta = \varphi(g_{w_1})h_{w_1}\cdots \varphi(g_{w_{\ell}})h_{w_{\ell}} \tau^\beta$.
    \end{itemize}
\end{definition}

    \item In order for the encryption and decryption to work, the assumptions that $\tau$ is central or of order $2$ are not necessary. The reason we work under these assumptions is, as explained in \cref{sec:homomorphic}, that in this case, the cryptosystem is unbounded additive homomorphic. 

    \item\label{rmk:wlog-gen}
Once Alice has fixed the elements $g_1,\dots,g_m$ and determined the public key, all computations inside $G$ actually take place inside the subgroup $\langle g_1,\dots, g_m \rangle$ that is generated by $g_1,\ldots,g_m$. So for cryptanalysis, we may and will assume that $G$ is generated by $g_1,\dots,g_m$, i.e.\ that  $G=\langle g_1,\dots, g_m \rangle$.

    \item The work of~\cite{Eprint} was inspired by~\cite{Baumslag}, which introduces the Learning Homomorphisms with Noise problem in order to construct a symmetric primitive. However, the noise accumulates in the construction of~\cite{Baumslag} in a manner akin to the error growth in LWE constructions. Leonardi and Ruiz-Lopez also introduce a symmetric primitive in~\cite{Eprint}, but we focus on the PKE construction in this work.

    \item The noise consists of the elements $h_1,\dots,h_m$ that are mixed into the product $h'$ in the second component $h$ of the ciphertext. These elements are chosen to be in the kernel of $\psi$ and therefore get erased during decryption. `Being contained in the kernel' of $\psi$ can thus be thought of as an equivalent of `the error being small' in the LWE-based encryption of~\cite{Regev} or `the noise being small' in LHN-based encryption of~\cite{Baumslag}. 
    The strength of Leonardi--Ruiz-Lopez encryption is that the noise does not accumulate and will not lead to systematic decryption errors, since in the decryption process, we can erase the noise neatly by applying~$\psi$. 

\end{enumerate}

\subsection{Cryptanalysis of the $\LHNPKE$ problem in the abelian case}\label{sec:assumptions}

In this paper we show that, if $G$ and $H$ are finite and abelian, under some reasonable assumptions that we introduce here, 
the $\LHNPKE$ problem for $G$ and $H$ can be reduced to the extended discrete logarithm problem (cf.\ \cref{def:eDLP}) and membership problem (cf.\ \cref{def:membership}) in some specific abelian $2$-groups\Mima{membership problem missing}\Chloe{fixed},
which leads to a polynomial-time attack in case the 2-part of $G$ is cyclic;
the general asymptotic complexity is given in our main theorem below.
We also attack the case with finite but nonabelian $G$ and abelian $H$, again under some reasonable assumptions, in Section~\ref{sec:limitations}.

In each section that follows, we explicitly mention under which assumptions from the following list we are working. 
For  a group $\Gamma$ and $\gamma_1,\ldots,\gamma_m\in\Gamma$, we are interested in the following properties:

\begin{enumerate}[label=$(\mathrm{A}\arabic*)$]
    \item\label{it:A1} $\Gamma$ is abelian;
    \item\label{it:A2} the largest positive odd factor of the order $|\Gamma|$ of $\Gamma$ is known (or easily computable); 
    \item\label{it:A3} $\gamma_1,\ldots,\gamma_m\in \Gamma$ are such that $\Gamma=\gen{\gamma_1}\oplus\ldots\oplus\gen{\gamma_m}$;
    \item\label{it:A4} The orders $|\gamma_1|,\ldots,|\gamma_m|$ of $\gamma_1,\ldots,\gamma_m$ are known (or easily computable).
\end{enumerate}

We discuss these assumptions -- why and when they are reasonable and/or needed -- in Sections \ref{RSAetc} (for \ref{it:A2} and \ref{it:A4}) and \ref{A3justify} (for \ref{it:A3}).
We can now state the main theorem of this paper:

\begin{theorem}\label{main-theorem}
\Annette{I am fine with putting in this theorem, as long as the statement is correct, but I am not too fond of it. I would prefer to just state it for abelian groups (or maybe even just for 2-groups) and then remark afterwards that under some assumptions (e.g. efficient computation in the quotient...) this can be generalized to certain non-abelian groups . It is very hard to read out of this new Theorem 3 what the statement is for the case of abelian groups (which was our target case). We should be more careful when we say things like just work in the quotient $\bar G$ in light of that computationally that might be non-trivial and involve more assumptions. So maybe it would be cleaner to just state and prove it for abelian groups and move the non-abelian stuff to a remark.  }
\Chloe{I have rewritten this taking into account Mima's comments: the old version (with Mima's comments) is still here but commented out. 
Re: only having it for abelian or 2-groups: I think that weakens the paper a lot, as this result is much stronger. In any real-life example it's hard to come up with cases where the cryptosystem would be efficient but quotienting out etc is not. I am also afraid that if we state it only for abelian groups, people will think the nonabelian case is secure (who don't read the whole paper)} \Mima{I see both points but I also think that if someone is really interested in the paper will read it thoroughly. We could also make sure to stress reductions. In any case, we either fix the problem with (A3) or we need to cut the theorem nonetheless. The easiest way to fix it is to only talk of 2-groups and then Annette's comment applies.}

%    Let $G$ be a finite group such that
%    $\overline{G} = G/[G,G] $ satisfies \ref{it:A2} and \ref{it:A3} \Mima{for which $\gamma_i$'s? It is also not really that (A3) is needed for G but rather for its Sylow 2-subgroup -- See Section 3.4. We can, however, to have a simpler phrasing, just assume $G=\gen{g_1}\oplus\ldots\oplus\gen{g_m}$ and then we will have it also for G2.},
%    and such that the
%    natural projection $G \rightarrow \overline{G}$ is efficiently computable.
%    Let $H$ be a finite group that satisfies \ref{it:A1}, \ref{it:A2}, and \ref{it:A3} \Mima{here, instead, it seems to me that we need any cyclic decomposition of $H$, not a prescribed one}.
%    Let $K$ be an arbitrary group.
%    Let $\xi$ and $\chi$ be probability distributions on $G$ and $H$ respectively.
%    Let the Sylow $2$-subgroup of $\overline{G}$, denoted $G_2$, be given in terms of a cyclic decomposition with $m$ summands.\Mima{This $m$ is really the number of the $g_i$'s: our proof strictly depends on that. We assume WLOG that $G$ is generated by the $g_i$'s for cryptanalysis but in this statement the $g_i$'s are not involved while I think they should be.}
%    Then there is a classical algorithm to solve $\LHNPKE$ for $G,H,K,\xi,$ and $\chi$
%    with asymptotic complexity
%    \[\textnormal{polylog}(|G|) + \textnormal{polylog}(|H|) + O\left(\frac{\log_2(|G_2|)}{m}2^{m/2}\right).\]

    Let $G = \langle g_1 \rangle \oplus \cdots \oplus \langle g_m \rangle$ be a finite group\Annette{This $G$ is abelian itself!} which satisfies \ref{it:A1}, \ref{it:A2}, and \ref{it:A3}.
    %and such that $\overline{G} = G/[G,G]$ admits efficient computation,
    %the natural projection $G \rightarrow \overline{G}$ is efficiently computable, and $\overline{G}$ satisfies \ref{it:A2}.
 \Mima{I am referring to my voice message from July 24th, 09:51: the bottom line is that (A3) is an assumption on $G$ ONLY together with the public information (more precisely the $g_i$'s used to encrypt the message). The theorem does not reflect this. And it might be that other $g_i'$s in $G$ do not satisfy these assumptions and then the theorem would be wrong. See also Lemma 8 and my addition there.}\Annette{I agree with Mima}
    Let $H$ be a finite group that satisfies \ref{it:A1} and \ref{it:A2} and for which we can efficiently compute a cyclic decomposition
    of its Sylow $2$-subgroup $H_2$.\footnote{Note that if we assumed \ref{it:A3} for $H$ together with a choice of generators, computing such a decomposition for $H_{2}$ would be easy, as $H$ satisfies \ref{it:A2}; that is, this requirement for $H_2$ together with \ref{it:A2} for $H$, is weaker than \ref{it:A3} for $H$.}
    Let $K$ be an arbitrary group.
    Let $\xi$ and $\chi$ be probability distributions on $G$ and $H$ respectively.  
    Then there is a classical algorithm to solve $\LHNPKE$ for $G,H,K,\xi,$ and $\chi$
    with asymptotic complexity
    \[\textnormal{polylog}(|G|) + \textnormal{polylog}(|H|) + O\left(\frac{\log_2(|G_2|)}{m}2^{m/2}\right),\]
    where $G_2$ denotes the Sylow $2$-subgroup of $G$.

    \Eleni{Given Mima's comments, we could rephrase: Let $G = \langle g_1 \rangle \oplus \cdots \oplus \langle g_m \rangle$ and $H = \langle u_1 \rangle \oplus \cdots \oplus \langle u_k \rangle$ be finite groups that both satisfy \ref{it:A1}, \ref{it:A2} and \ref{it:A3}, with the given set of generators (\ref{it:A1} perhaps for $\overline{G}$ instead of $G$). Then, there is a classical algorithm ...} \Mima{$g_1,\ldots, g_m$ ($m$!) are important for $G$ (in the reduction you see all details of how), for $H$ we do not care what the decomposition is as long as we have one.}\Eleni{Yes, m. I wrote the decomposition for $H$ as well because I think we we need it, if we are to employ {BerVega14}.} \Mima{same as I was saying - my point was about the $g_i$'s.} 
    \Chloe{I believe that what is written, `and for which we can efficiently compute a cyclic decomposition
    of its Sylow $2$-subgroup $H_2$', is sufficient for BerVega14}
    \Mima{I don't want to make suggestions of how to rephrase it as different people have different opinions on what we should do (if we weaken the statement we can be precise, if not, it becomes a very technical statement) and I think we should talk about it.}
\end{theorem}

\noindent
We will give the proof of this theorem in Section~\ref{main-theorem-proof}.

Note in particular that if $m=1$, that is if $G$ is cyclic, this attack  algorithm is  polynomial-time (in fact, as we will see, it is sufficient for $G_2$ to be cyclic).
In Section~\ref{sec:limitations}, we discuss how to attack instances with nonabelian $G$, assuming that we can work with the abelianization $G/[G,G]$ of $G$.
In Section~\ref{finiteabelianquantum}, we also briefly discuss a polynomial-time quantum algorithm, complementing the existing literature for the torsion-free case~\cite{Eprint}.
In Section~\ref{sectionConvert}, we discuss instantiations in which $H$ is not abelian (i.e. does not satisfy \ref{it:A1}) but which reduce to the case in which $H$ is abelian.
In Section~\ref{sec:solvable}, we discuss how one could develop a quantum algorithm for the solvable case.

\section{Simple instantiations and security}\label{sec:2}

In this section, we describe some simple instantiations of Leonardi--Ruiz-Lopez encryption. 
The abelian case is a central focus of this paper as it is much simpler to describe than the general case;
the description below is for the reader who wishes only to understand the abelian case.
We also describe the noiseless case in order to highlight the role that the noise plays in the encryption.
Finally we discuss the requirements on the setup parameters of Leonardi--Ruiz-Lopez encryption in order to achieve security against some naive classical attacks, 
concluding this section with a list of properties that the groups must have in order to avoid the trivial classical attacks outlined in this section.

\subsection{The abelian case} \label{sec:abelian}
If $H$ is abelian, we can rewrite %the message 
$(g,h)= (g_{w_1}\cdots g_{w_{\ell}},\varphi(g_{w_1})h_{w_1}\cdots \varphi(g_{w_{\ell}})h_{w_{\ell}}\cdot\tau^\beta )$ as
\[ (g,h)= (g_{w_1}\cdots g_{w_{\ell}},\varphi(g_{w_1})\cdots \varphi(g_{w_{\ell}})h_{w_1}\cdots h_{w_{\ell}}\cdot\tau^\beta )=(g,\varphi(g)h_{w_1}\cdots h_{w_{\ell}}\cdot\tau^\beta).\] 
If both $G$ and $H$ are abelian, it makes sense to switch to the following notation: 
Instead of choosing indices $w_1,\dots,w_{\ell}$, Bob just chooses non-negative integers $r_1,\dots,r_m$ and encrypts $\beta$ to 
\[ (g,h)= (g_1^{r_1}\cdots g_m^{r_m},\varphi(g_1^{r_1})\cdots \varphi(g_m^{r_m})h_1^{r_1}\cdots h_m^{r_m}\cdot\tau^\beta )=(g,\varphi(g)h_1^{r_1}\cdots h_m^{r_m}\cdot\tau^\beta).\] 
%\Vishnu{last term to $(g, \varphi(g)h'\tau^{\beta})$}
This system has been claimed to not be quantum secure in \cite[Section~8.2]{Eprint}, cf.\ also \cref{sec:LRL-cryptanalysis}, while we discussed security in the classical sense in \cref{sec:assumptions}. Moreover, in \cref{sectionConvert} 
we describe instantiations of the $\LHNPKE$ problem in which $G$ and $H$ are nonabelian but the security reduces to the case in which they are.

\subsection{The noiseless case}\label{sec:noiseless}
Let us assume that $h_1=h_2=\dots=h_m=1$. Then the public key consists of all pairs $(g_i,\varphi(g_i))$ together with $\tau$ and Bob would encrypt the message $\beta \in \{0,1\}$ to  
\[ (g,h)=(g_{w_1}\cdots g_{w_{\ell}},\varphi(g_{w_1})\cdots \varphi(g_{w_{\ell}})\tau^\beta)=(g,\varphi(g)\tau^\beta).
\] 
Let Eve be an attacker who is aware of the fact that Alice decided to work in a noiseless setting. Then Eve knows all $g_i$'s as well as their images $\varphi(g_i)$ from the public key. If she can write $g$ as a product in $g_1,\dots,g_m$, then she can compute $\varphi(g)$. Knowing $\tau$ from the public key, Eve can then decrypt $(g,h)=(g,\varphi(g)\tau^\beta)$. Note that even if Eve did not use the same word $w_1\cdots w_{\ell}$ as Bob to write $g$ as a product in $g_1,\dots,g_m$, she would nonetheless obtain the correct value of $\varphi(g)$. 

Of course, finding such a word might still be a hard problem, even if $m=1$. For example, if $G$ is the multiplicative group of a finite field and $g_1$ is a generator, finding a word in $g_1$ defining $g$ is the same as solving the discrete logarithm problem, which is known to be hard for classical computers (though there are quantum algorithms to solve it, see \cite{Shor_1, Shor_2}). Alice, on the other hand, will probably have a closed form describing $\varphi$ that does not require to write elements as products in $g_1,\dots,g_m$ when she applies $\varphi$ in the decryption process (and similarly for $\psi$). In the $m=1$ example above, she would choose $\varphi$ to take every element to a certain power. 

For attacking the cryptosystem in the general case, a possible strategy is to construct attacks that reduce to the noiseless case. We will come back to such attacks in \cref{sec:security}.

\subsection{Security}\label{sec:security}

In all that follows, let $\lambda$ be the security parameter.\footnote{
Typically, we want any computations undertaken by the user to have complexity that is polynomial in $\lambda$, and an attacker who attempts to decrypt by guessing any unknowns should only succeed with probability at most $2^{-\lambda}$.}

\subsubsection{Many homomorphisms}
First of all, note that an attacker who can guess both $\varphi$ and $\psi$ can decrypt in the same way Alice does. 
Denote the set of all possible choices for $\psi$ by
\[\operatorname{Hom}(H,K)^- = \{ \psi \in \operatorname{Hom}(H,K) :
\mathrm{Z}(H) \not\subseteq \ker(\psi)\}.\]
To avert brute force attacks, the groups $G,H,K$ should be chosen in such a way that $\operatorname{Hom}(G,H)$ and $\operatorname{Hom}(H,K)^-$ are of size at least $\Theta(2^{\lambda})$,\footnote{
We are using Bachmann-Landau notation for complexity, see for example Section 1.2.11.1 of~\cite{Knuth}.
} 
and $\varphi$ and $\psi$ should be sampled uniformly at random from
$\operatorname{Hom}(G,H)$ and $\operatorname{Hom}(H,K)^-$ respectively. 
This ensures that if an attacker guesses $\varphi$ she succeeds with probability $2^{-\lambda}$, and similarly for $\psi$.

\subsubsection{Words}
As we already saw in the noiseless case, there are links between the $\LHNPKE$ problem and the ability of an attacker to write $g$ as an expression in the generators $g_1,\ldots,g_m$.
Assume for instance that Eve wants to decrypt the cyphertext $(g,h)$. She knows that $g$ is a product in $g_1,\dots,g_m$ and recall that $g_1,\dots,g_m$ are public. 
If she knows the exact expression of $g$ as a product in $g_1,\dots,g_m$ that Bob used in the encryption process, then she can compute $h'$ from the public key, erase it from $h=h'\tau^\beta$, and recover the message $\beta$. 
It is important to note that other than in the noiseless case, it is in general not enough to find \textit{any} expression of $g$ as a product in $g_1,\dots,g_m$, because that product will in general not produce the correct term $h'$ yielding to a different accumulation of the noise. 
That is, the attacker needs to recover the \emph{correct} word $w$, not just any expression of $g$ in $g_1\ldots,g_m$.
To see this, recall that $(g,h)=(g,h'\tau^\beta)$ with
$$(g,h')= (g_{w_1}\cdots g_{w_{\ell}},\varphi(g_{w_1})h_{w_1}\cdots \varphi(g_{w_{\ell}})h_{w_{\ell}})$$ and assume that the attacker found some word $\tilde w_1\dots \tilde w_{\tilde \ell}$ such that $g=g_{\tilde w_1}\cdots g_{\tilde w_{\tilde \ell}}$. Using the public elements $(g_1,\varphi(g_1)h_1),\dots,(g_m,\varphi(g_m)h_m)$, it is then possible to compute an element $
\tilde h'=\varphi(g_{\tilde w_1})h_{\tilde w_1}\cdots \varphi(g_{\tilde w_{\tilde \ell}})h_{\tilde w_{\tilde \ell}}$ but in general $\tilde h'\neq h'$ and the knowledge of $\tilde h'$ together with $h$ does not help to compute $\tau^\beta$. 
In the cryptanalysis we carry out in \cref{ab-groups} for finite abelian $G$, $H$, and $K$, we give a reduction of $\LHNPKE$ to the extended discrete logarithm problem for finite $2$-groups; our reduction circumnavigates the issue of finding the correct word.

\subsubsection{An attack on instances with few normal subgroups}

The idea behind the following attack is to replace Alice's secret $\psi:H\rightarrow K$ with some new $\bar \psi:H\rightarrow L$ erasing the noise without erasing~$\tau^\beta$ (for instance $L$ could be a quotient of $H$, as described below). 
A similar attack was also described in Section 7.2 of \cite{Eprint}. 
Assume that Eve knows a normal subgroup $N$ of $H$ that contains all elements $h_1,\dots,h_m$ but does not contain $\tau$. She can then define $\bar \psi\colon H \to H/N$ as the natural projection and by applying $\bar \psi$ to all second coordinates of the elements $(g_i,\varphi(g_i)h_i)$ in the public key and to the second coordinate of the encrypted message $(g,h)$ she can switch to the noiseless case; cf.\ \cref{sec:noiseless}. %In particular, she can decrypt any ciphertext if she can solve the eDLP in $G$. 
We deduce in particular, that there should be at least $\Theta(2^{\lambda})$ normal subgroups in $H$, 
so that if an attacker guesses $\ker(\psi)$ she succeeds with probability $2^{-\lambda}$.

\subsubsection{An attack on instances with weak normal subgroups}

Now suppose that an attacker can find a normal subgroup $N$ of $H$ that contains $\varphi(g_i)h_i$ for all $i=1,\dots,m$ but does not contain $\tau$ (note that these elements are all public). 
If efficiently computable, she can then directly apply the projection $H\to H/N$ to the second coordinate in the encrypted message $(g,h)$ and can deduce that $\beta$ equals zero if and only if she obtained the neutral element in $H/N$. 
To avoid such an attack, after the key generation process,  Alice must check (assuming that she is capable of it) whether the normal closure of $\langle \varphi(g_1)h_1,\dots,\varphi(g_m)h_m \rangle$ contains $\tau$ and, if it doesn't, she should choose a different key. 
For cryptanalysis, we may thus assume that $H$ equals the normal closure of $\langle \varphi(g_1)h_1,\dots,\varphi(g_m)h_m \rangle$ (as otherwise, we can just work in this smaller group). In particular, if $H$ is abelian, we may assume that $H=\varphi(G)\ker\psi$. 
%\Vishnu{Depending on the outcome of the discussion on Prop 5, we can add a line about the subgroup membership of $\tau$ in the normal closure?}\Mima{maybe it is not that necessary after all} \Vishnu{Ok.}

\subsubsection{A summary of the discussed security assumptions}\label{sec:security-assumptions}

We conclude this section with a list of necessary properties for security deduced from the list of naive attacks above:

\begin{enumerate}[label=$(\mathrm{S}\arabic*)$]
    \item\label{it:S1} $\operatorname{Hom}(G,H)$ and $\operatorname{Hom}(H,K)^-$ are of size exponential in the security parameter;
    % \item the implementation for this cryptosystem does not use unique words to describe elements in terms of their generators; \Mima{This is now in contradiction with what we do in the abelian case or maybe not exactly, but we should explain better what we mean $ri$ vs $r_i\bmod |g_i|$} \Annette{we could just rephrase and say that there need to be many different ways to write group elements in $G$ as products and that Alice and Bob cannot use some canonical/unique presentation}
    \item\label{it:S2} finding the precise word $w$ used to express $g$ as a product in the $g_i$'s in the encryption phase has complexity that is exponential in the security parameter;
    \item\label{it:S3} the number of normal subgroups in H is exponential in the security parameter;
    \item\label{it:S4} the normal closure of $\langle \varphi(g_1)h_1,\dots,\varphi(g_m)h_m \rangle$ contains $\tau$.
\end{enumerate}

\section{Cryptanalysis in the finite abelian case}\label{abelian-case} 

In this section, we prove~\cref{main-theorem} via a series of reductions of the $\LHNPKE$ problem.
Recall that~\cref{main-theorem} assumes that $G$ and $H$ are finite and abelian (in addition to a few other assumptions).
In Section~\ref{sec:Habelian}, we will show that this can be taken one step further by giving a reduction from $G$ being finite and $H$ being finite and abelian to both $G$ and $H$ being finite and abelian.
In Section~\ref{sec:reduce_2}, we will give a reduction from $G$ and $H$ being finite and abelian to $G$ and $H$ being finite abelian 2-groups (under some reasonable assumptions).
In Section~\ref{sec:attack2gps}, we will give a reduction from $\LHNPKE$ for finite abelian $G$ and $H$ to the eDLP problem\footnote{See \cref{def:eDLP} for a formal definition.} 
for $G$ and the generator membership problem\footnote{See \cref{def:membership} for a formal definition.} 
for $H$ (under some reasonable assumptions).

In \cref{main-theorem-proof} we give the proof of~\cref{main-theorem} and in \cref{sec:limitations} we discuss the genericity and limitations of the assumptions made in \cref{sec:assumptions}, as well as the impact of the reductions made. 

Leonardi and Ruiz-Lopez~\cite{Eprint} dismissed the abelian instantiation due to an argument that there should exist a polynomial-time quantum algorithm for the $\LHNPKE$ problem;
the reduction is more complex than is suggested in~\cite{Eprint} but their statement is true as we show in \cref{sec:LRL-cryptanalysis}. 
Nevertheless, the unbounded homomorphic property of the proposed cryptosystem is sufficiently powerful that a classically secure construction would also be of great interest to the cryptographic community.

\subsection{A reduction of $\LHNPKE$ with abelian noise to abelian groups}\label{sec:Habelian}

In this section we take $G$ to be finite and $H$ to be finite and abelian. 
We show that, when the natural projection $G\rightarrow \overline{G}=G/[G,G]$ from $G$ to its abelianization is efficiently computable, then, for cryptanalysis, one can consider $\overline{G}$ instead of $G$, with the distribution $\overline{\xi}$ that is induced on $\overline{G}$ by the projection. 

\begin{lemma}\label{G_also_abelian}
Assume $H$ is abelian. Then the $\LHNPKE$ problem for $G,H,K,\xi,$ and $\chi$ is at most as hard as the $\LHNPKE$ problem for $\overline{G} = G/[G,G], H, K, \overline{\xi}$, and $\chi$.
\end{lemma}
\begin{proof}
Since $H$ is abelian, the commutator subgroup $[G,G]$ of $G$ is contained in the kernel of $\varphi$. Define $\overline{G}=G/[G,G]$. Then $\varphi\colon G\to H$ induces a well-defined homomorphism $\overline{\varphi} \colon \overline{G} \to H$ and any ciphertext $(g,h)$ can be interpreted as the ciphertext $(\bar g, h)$ in the cryptosystem given by $\overline{\varphi} \colon \overline{G} \to H$, $\psi\colon H \to K$ as before and public key given by $(\overline{g_i},\bar \varphi(\overline{g_i})h_i)$ for $i=1,\dots,m$ together with the same $\tau$ as before. Indeed, if Bob used the word $w_1\cdots w_{\ell}$ to encrypt $(g,h)$, then the same word gives rise to the encryption $(\bar g, h)$ in the new cryptosystem of the same message. 
\qed
\end{proof}

\subsection{Reductions of $\LHNPKE$ in the finite abelian case}\label{ab-groups}

In \cref{ab-groups} we discuss the $\LHNPKE$ problem for $G$ and $H$ finite and abelian. 
\Chloe{Do we actually care whether or not $K$ is finite or abelian?}\Mima{No, but it will reduce to that so it is no harm to assume it}\Chloe{We don't explicitly reduce to that anywhere I think}\Mima{See my comment in 3.1}\Mima{I am not sure where my comment went but it essentially said that what you care about in $K$ is just the image of $\psi$ and that one will be abelian if $H$ is}

\subsubsection{From abelian groups to $2$-groups}\label{sec:reduce_2}
In this section, we take $G$ and $H$ to be finite groups for which \ref{it:A1} and \ref{it:A2} hold,
i.e.\ we assume that both $G$ and $H$ are finite abelian groups and that the odd parts of $|G|$ and $|H|$ are known. We let $G_2$ and $H_2$ denote the unique Sylow $2$-subgroups of $G$ and $H$, respectively,
and $\xi_2$ and $\chi_2$ denote the distributions $\xi$ and $\chi$ restricted to $G_2$ and $H_2$ respectively.
In the new language, the elements $\tau$ and $\tau^\beta$ belong to $H_2$.
In the following, we show that in the case of abelian groups, the $\LHNPKE$ problem for the original data reduces to the $\LHNPKE$ problem for the 2-parts $G_2$ and $H_2$ (with the induced data). %, yielding a classical attack.

\begin{lemma}\label{result_abelian}
Suppose that $G$ and $H$ are finite groups satisfying \ref{it:A1} and \ref{it:A2}. Then the $\LHNPKE$ problem for $G, H, K, \xi$, and $\chi$ is at most as hard as the $\LHNPKE$ problem for $G_2, H_2, K, \xi_2$, and~$\chi_2$.
%Assume that $H$ is abelian. 
\end{lemma}
\begin{proof}
Write the orders of $G$ and $H$ as
\[ |G|=2^{n_G}q_G, \  |H|=2^{n_H}q_H, \] 
where $q_G$ and $q_H$ are odd numbers. 
Let $q$ be a common multiple of $q_G$ and $q_H$ that can be computed efficiently.
Then for any elements $g \in G$ and $h \in H$, 
the orders of $g^q$ and $h^q$ are both powers of $2$, i.e.\ $g^q\in G_2$ and $h^q\in H_2$. 
Moreover, $G$ and $H$ being abelian, the assignment $x \rightarrow x^q$ defines  homomorphisms $G\to G_2$ and $H\to H_2$.

Equip $G_2$ and $H_2$ with the induced distributions $\xi_2$ and $\chi_2$ and write $\varphi_2\colon G_2\to H_2$ and $\psi_2\colon H_2\to K$ for the restrictions of $\varphi$ and $\psi$ to $G_2$ and $H_2$, respectively. Assume that $\beta$ can be recovered from $(G_2, H_2, K, \xi_2,\chi_2)$. We show that $\beta$ can be determined from $(G, H, K, \xi,\chi)$.

Recall that the pair $(g,h)=(g_1^{r_1}\cdots g_m^{r_m},\varphi(g)h_1^{r_1}\cdots h_m^{r_m}\cdot\tau^\beta)=(g,h'\tau^\beta) \in G\times H$ is public. Raising both entries to their $q$-th power, one obtains 
 \[
 (g^q,h^q)=((g_1^q)^{r_1}\cdots (g_m^q)^{r_m},\varphi(g^q)\cdot(h_1^q)^{r_1}\cdots (h_m^q)^{r_m}\cdot(\tau^q)^\beta)=(g^q,(h')^q\tau^\beta)\in G_2\times H_2.
 \] 
From the last equation it is clear that 
the elements $g_1,\dots,g_m \in G$ and 
$h_1,\dots,h_m \in H$ are replaced with their $q$-th powers $g_1^q,\dots,g_m^q \in G_2$ and $h_1^q,\dots,h_m^q\in H_2$. Moreover, as $\tau=\tau^q$, it holds that $\tau \in H_2\setminus\ker(\psi_2)$ and so the pair $(\tau,\beta)$ is preserved. By assumption $\beta$ is determined from the data associated to the $2$-parts and so the proof is complete.
\qed
\end{proof}

\subsubsection{From $\LHNPKE$ to eDLP and membership}\label{sec:attack2gps}

\noindent
In this section, we work under Assumption \ref{it:A1} of \cref{sec:assumptions} for $\Gamma\in\{G,H\}$ and Assumptions \ref{it:A3} and \ref{it:A4} of \cref{sec:assumptions} for $\Gamma=G$ and $\{\gamma_1,\ldots,\gamma_m\}=\{g_1,\ldots,g_m\}$. 
Under Assumptions \ref{it:A1}, \ref{it:A3}, and \ref{it:A4} the following problem is well-posed.

\begin{definition}\label{def:eDLP}
The \emph{extended discrete logarithm problem (eDLP)}  for the finite abelian group $\Gamma=\gen{\gamma_1}\oplus\ldots\oplus\gen{\gamma_m}$ is the problem of determining, for each $x\in \Gamma$, the unique vector \[(\alpha_1,\ldots,\alpha_m)\in \Z/|\langle\gamma_1\rangle|\Z\times \ldots\times\Z/|\langle\gamma_m\rangle|\Z\] such that 
    \[x=\gamma_1^{\alpha_1}\cdots \gamma_m^{\alpha_m}.\]
\end{definition}

\noindent
Note that eDLP is just called \emph{discrete logarithm problem} (DLP) in \cite{Sutherland}. For a discussion of existing algorithms to solve it, we refer to \cref{sec:eDLP2gps}. 
In the following, we additionally define a natural variation of the classical membership problem for groups.

\begin{definition}\label{def:membership}
    Let $\Gamma$ be a group. The \emph{generator membership problem} for $\Gamma$ is the problem of determining whether, given  a finite subset $Y$ of $\Gamma$ and an element $x\in \Gamma$, one has $x\in\gen{Y}$.  
\end{definition}

\begin{lemma}\label{lem:h_i<g_i}
Assume that $G$ and $H$ are both finite and abelian and that $G = \langle g_1 \rangle \oplus \cdots \oplus \langle g_m \rangle$ \Mima{together with $g_1,\ldots,g_m$} additionally satisfies \ref{it:A3} and \ref{it:A4}. Then, the $\LHNPKE$ problem for $G,H,K,\xi,$ and $\chi$ is at most as hard as the hardest of the \emph{eDLP} problem for $G$ and the generator membership problem for $H$. \Mima{Why is it phrased like this now? They are not always comparable. I suggest rephrasing as 'at most as hard as the hardest'.. unless we know for sure that the crypto community will read it like this.}
\end{lemma}
 
\begin{proof}
Let us write $G$ as $G=\gen{g_1}\oplus\ldots\oplus\gen{g_m}$. Given the public pair $(g,h)=(g,h'\tau^\beta) \in G \times H$ as in \cref{sec:cryptosystem}, let $r_1\ldots,r_m\in\Z$ be the non-negative integers chosen to write $$g=g_1^{r_1}\cdots g_m^{r_m} \ \ \textup{ and } \ \ h=\varphi(g)h_1^{r_1}\cdots h_m^{r_m}\cdot\tau^\beta.$$
For each $i\in\{1,\ldots,m\}$, set $\ell_i=\varphi(g_i)h_i$ and, given that the order of $\varphi(g_i)$ divides $|\langle g_i \rangle|$, compute
 \[
 \ell_i^{|\langle g_i \rangle|}=(\varphi(g_i)h_i)^{|\langle g_i \rangle |}=h_i^{|\langle g_i \rangle|}.
 \]
 As a consequence, the subgroup $M$ of $H$ that is generated by $Y=\{\ell_i^{|\langle g_i\rangle |} : i=1,\ldots,m\}$ is contained in $\ker(\psi)$, and thus $M$ does not contain $\tau$. %Since the orders of the $g_i$'s are known, the subgroup $M$ can be easily determined.

Assume that eDLP is solvable in $G$. With $h$ being as above, let $s_1,\ldots,s_m$ be such that, for any choice of $i$, one has $s_i\equiv r_i \bmod |\langle g_i \rangle|$. Denote by $t_{i}$ the integers satisfying $r_{i} - s_{i} = |\langle g_{i} \rangle|t_{i}$ and
\begin{align*}%\label{eq:X}
    x= \prod_{i=1}^m(\ell_i^{s_i})^{-1}\cdot h=\prod_{i=1}^m\ell_i^{r_i-s_i}\cdot\tau^{\beta} = \prod_{i=1}^m(\ell_{i}^{|\langle g_{i}\rangle|})^{t_{i}}\cdot \tau^{\beta} \eqqcolon h_{M}\cdot \tau^{\beta},
\end{align*}
where we easily observe that $h_{M} \in M$. We now see that $\beta = 0$ if and only if $x\in M$ and thus, if we are able to solve the generator membership problem for $H$, we are also able to determine whether $\beta=0$ or $\beta=1$. 

%We see that $Y$ must be of polynomial-size since it must be smaller than the public key. Then, by \cite[Lemma~3(a)]{BerVega14}, 
%\Chloe{`as $H$ satisfies (A1), (A3), and (A4')?}
%deciding whether $x \in M$ or not is $O\Big{(}\text{polylog}(|H|)\Big{)}$. 

% Another algorithm that decides the membership problem, is given in an older paper by Iliopoulos \cite[Prop. 6.2, pg.91]{Iliop}. Its complexity depends on $|M|$ and $|x|$, instead. 
%Let us also remark that for this algorithm, neither assumption \ref{it:A3} nor assumption \ref{it:A4} are required for our group $H$; however, we still need to be able to compute $|x|$ and $|\ell_i^{|g_i|}|$ for each $i = 1, ..., m$. The complexity of this algorithm, adjusted to our notation, is given as: 
%$$O\Big{(}|x|^{1/2}(\log|x| + |M|) + m|M|^{1/2}\log^{3}|M|(\log|M| + \xi)\Big{)},$$ 
%where $\xi$ is the number of elementary operations required for a group operation. 
%\Chloe{This is exponential complexity. The order of $x$ is about $|H|$.}
%As given in the same paper, a group operation in our case requires at most
%$$\Big{(}c|M|\log|M|\log\log|M|\Big{)} \text{-many elementary operations,}$$ where $c$ is some fixed positive constant (please see \cite[pg. 67]{Iliop} about $\xi$, and \cite[Theorem 3.1. pg.4]{Iliop} about the number of elementary operations). 

\Chloe{I have commented out the membership problem stuff to put in the proof of the main theorem instead (along with the eDLP complexity). The motivation behind this is to have this lemma be "reductions to known problems" -- since we have stated in in general for abelian groups and we really want things for 2-groups only which is much easier.}
\qed
\end{proof}

\subsection{Proof of Theorem~\ref{main-theorem}}\label{main-theorem-proof}

\noindent
We are given the following public information:
\begin{enumerate}
    \item $\{(g_1,\varphi(g_1)h_1),\ldots,(g_m,\varphi(g_m)h_m),\tau\}$ and 
    \item $(g,h)=(g,h'\tau^\beta) \in G \times H$
\end{enumerate}
One can proceed as follows: \Mima{this whole proof is basically all contained in 3.4 already, in a precise way for what concerns the reductions}\Eleni{Yes. And it is very nicely written for general $G$, so why not keep $G$ as it is then? I.e. not necessarily abelian, and then we  pass to $\overline{G}$. We also need \ref{it:A3} for $H$, to be able to employ {BerVega14}.} \Mima{The point of writing it in a section (at least this is why I did it like that) was that we could as each step write what was needed. If we want to write it as a theorem it becomes something unreadable if we keep it as is (without weakening anything I mean). That is my point of view now but I am happy to be proven wrong}
\begin{itemize}
%    \item 
%    In case $G$ is not abelian, 
%    the $\LHNPKE$ problem for $G,H,K,\xi,$ and $\chi$ is reduced to the $\LHNPKE$ problem for $\overline{G}, H, K, \overline{\xi}$, and $\overline{\chi}$ following \cref{G_also_abelian}; note that $\overline{G}$ is abelian.
    \item By \cref{result_abelian}, as \ref{it:A1} and \ref{it:A2} are both satisfied for both $G$ and $H$,
    the $\LHNPKE$ problem for $G, H, K, \xi,$ and $\chi$ is at most as hard as the $\LHNPKE$ problem for
    $G_2, H_2, K, {\xi}_2$, and ${\chi}_2$, where
    $G_2$ and $H_2$ are the Sylow 2-subgroups of $G$ and $H$, respectively. \Eleni{... where $G_{2} = \langle {g_1}^{q_{G}} \rangle \oplus \cdots \oplus \langle g_m^{q_{G}} \rangle$ and $H_{2} = \langle u_1^{q_{H}} \rangle \oplus \cdots \oplus \langle u_k^{q_{H}} \rangle$, and the largest odd factors $q_{G}, q_{H}$ of the orders of $G$ and $H$ respectively, are known by \ref{it:A2}.}
    \Chloe{I have put the $G$ parts of this into the next comment. I don't think \ref{it:A3} is quite the right thing for $H$ which is why that was written out instead}
    \item As $G = \langle g_1 \rangle \oplus \cdots \oplus \langle g_m \rangle$ satisfies \ref{it:A3}, and by \ref{it:A2} the largest odd factor $q_G$ of $|G|$ is known or easily computable, we also have that $G_2 = \langle g_1^{q_G} \rangle \oplus \cdots \oplus \langle g_{m}^{q_G} \rangle$ satisfies \ref{it:A3};
    note that some components may now be trivial, but they are necessarily still distinct as no power of $g_i$ was contained in $\langle g_j \rangle$ for $i\neq j$.% \Vishnu{$\#$ is not used anywhere else, so changing it to $|G|$.}
    \Chloe{\ref{it:A3} doesn't follow from this. But why do the components need to be distinct?} \Annette{To me it seems that the components are still 'distinct' (as Eleni wrote), i.e. the subgroups $\langle g_i^q\rangle $ also form a direct sum (since the groups $\langle g_i \rangle $ do. It might very well happen that some of the new components $\langle g_i^q\rangle $ are trivial, but that should not cause any problems in the proof. }\Chloe{ you are right: this is ok, hooray! I've added a comment about this}
    Also, as $G_2$ is a 2-group, checking the order of any element is polynomial-time,
    just by repeated squaring and checking for the identity, so \ref{it:A4} is also satisfied.\Mima{Where is the (A3) assumption coming in here?}\Chloe{not until we invoke the lemma, but talking about A4 doesn't make sense without A3. Have clarified (hopefully)} \Mima{This is sadly not true as the theorem is written now -- see my comments there}\Eleni{But if we stick to the fixed set of generators, for both $G$ and $H$ (with A1,A2,A3, satisfied for both $\overline{G}$ and $H$), and then do the reductions on these generators to get to the $2$-part (as it is done in Section 3.4), then we are ok, no? This is all I was trying to say with my comments.}\Mima{we might be saying the same? the important thing is to fix the public information beforehand - we can't do it 'for any' choice -- that is the main point I was trying to make. so if we write a statement including distributions etc, then we should also reintroduce the public information}
    Therefore, by \cref{lem:h_i<g_i},
    the $\LHNPKE$ problem for $G_2,H_2,K,\xi_2,$ and $\chi_2$ is reduced to the eDLP problem in $G_2$ and the generator membership problem in $H_2$.
\end{itemize}
The complexity of the eDLP problem in a 2-group $G_2$ with $m$ cyclic components is given by
\[ O\left(
\frac{\log\log|G_2|}{\log\log\log|G_2|}\log|G_2|+\frac{\log_2|G_2|}{m}2^{m/2}
\right)
\]
by \cite[Cor.~1]{Sutherland} (see Section~\ref{sec:eDLP2gps} for more details).
The membership problem for $H_2$ has polylogarithmic complexity by \cite[Lemma 3(a)]{BerVega14};\Eleni{Why is the comment on the required decomposition from {BerVega14}, commented out? I think it is important.}\Mima{no idea why}\Chloe{because the required decomposition is now an assumption in the theorem instead -- I wasn't sure if \ref{it:A3} and \ref{it:A4} were quite the right assumptions, I'll rephrase it instead}
note that the required decomposition of $H_2$ discussed in \cite[Section 3.1]{BerVega14}
exists by assumption.\Eleni{Now I see the new assumption in Theorem3 for $H_{2}$. What is the complexity for computing this decomposition for $2$-groups? When I was looking for a general group, it was exponential. Perhaps it is more `doable' (and `believable') to assume \ref{it:A3} for $H = \langle u_1 \rangle \oplus \cdots \oplus \langle u_k \rangle$, get the corresponding set of generators for $H_{2} = \langle u_1^{q_{H}} \rangle \oplus \cdots \oplus \langle u_k^{q_{H}} \rangle$, and compute the orders of the $u_{i}^{q_{|H|}}$ (which we can do in polynomial-time since we are in a $2$-group). I believe that we can now acquire the required decomposition of {BerVega} for $H_{2}$, assuming we can compute efficiently $\text{lcm}$'s of the orders of the $u_{i}^{q_{|H|}}$'s (?).}
\Chloe{We could make a remark about $H$ satisfying \ref{it:A3} being an example of an $H$ in which computing this decomposition is easy. But it seems restrictive and not necessarily easier to read to impose this restriction in the theorem in my opinion.}
%is polylogarithmic time for a 2-group satisfying \ref{it:A1} and \ref{it:A3}.
\qed

\subsection{Assumptions and reductions}\label{sec:limitations}

Assume in this section that $G$ and $H$ are finite and that $H$ is abelian. One can then, as before under some reasonable assumptions, reduce to the abelian case of \cref{main-theorem} as follows.

Thanks to \ref{rmk:wlog-gen}, for cryptanalysis purposes, we can replace $G$ with $\tilde{G}=\gen{g_1,\ldots,g_m}$ and we do so. Moreover, 
in view of \cref{G_also_abelian}, the $\LHNPKE$ problem on $(\tilde{G},H,K,\xi,\chi)$ is reduced to the $\LHNPKE$ problem on $(\overline{G},H, K, \overline{\xi},\chi)$, where $\overline{G}$ denotes the abelianization of $\tilde{G}$ and $\overline{\xi}$ is as in \cref{sec:Habelian}. Assumption \ref{it:A1} holds for $\overline{G}$, 
the elements $g_1,\ldots,g_m\in G$ are replaced with their images $\overline{g_1}, \ldots,\overline{g_m}$ in $\overline{G}$, 
and the homomorphism $\phi:\tilde{G}\rightarrow H$ is replaced with the induced homomorphism $\overline{G}\rightarrow H$, which we identify with $\phi$, for simplicity.

Assuming now that \ref{it:A2} holds for $H$ and $\overline{G}$, \cref{result_abelian} allows to reduce the $\LHNPKE$ problem on $(\overline{G},H)$ to the $\LHNPKE$ problem on $(\overline{G}_2,H_2)$, where $\overline{G}_2$ and $H_2$ denote the Sylow $2$-subgroups of $\overline{G}$ and $H$. Here the elements $\overline{g_1},\ldots,\overline{g_m}$ and $h_1,\ldots,h_m$ are replaced by 
\[\overline{g_1}^q,\ldots,\overline{g_m}^q \ \ \textup{ and } \ \ h_1^q,\ldots,h_m^q\]
where $q$ denotes the least common multiple of the odd parts of $|\overline{G}|$ and $|H|$. 

Note that $\overline{G}_2=\gen{\overline{g_1}^q,\ldots,\overline{g_m}^q}$ \Chloe{this is $\overline{G}_2$} and  suppose, at last, that \ref{it:A3}\Chloe{this is way too strong, wanting none of the $\overline{g_i}^q$s to be identified in all the quotienting, assuming this is meant to read that \ref{it:A3} together with $\overline{g}_1^q,\ldots,\overline{g}_m^q$ holds. We should say instead something like, `suppose that, after reordering, there is a $m' \leq m$ for which \ref{it:A3} holds for $G_2$ together with $\overline{g_1}^q,\ldots,\overline{g_{m'}}^q$}\Annette{I see your point that in practice, it might be unlikely for this assumption to hold, but at least this way everything that is written is correct. It is not that simple with the reordering, because remember that there are also all the $h_i$'s involved. E.g. if $\bar g_1=\bar g_2$ then still in the message you have both $h_1$ and $h_2$ mixed in, so following Lemma 4, we need both $\bar g_1$ and $\bar g_2$.}\Chloe{You are right about the $h$s. We use \ref{it:A3} for Lemma 8, and I can't see a problem with the proof of Lemma 8 if some of the components are identified -- but would appreciate a second pair of eyes to look at this. If we make such a (hopefully unnecessary) assumption on $G$ then we can't claim that our assumptions are mild / apply to all the relevant cases and it really undermines the point of the paper, so I think it's important to remove this if we can (it also applies to the proof of Theorem~\ref{main-theorem}).} 
holds for $\overline{G}_2$ together with $\overline{g_1}^q,\ldots,\overline{g_m}^q$;
that is 
\[\overline{G}_2=\gen{\overline{g_1}^q}\oplus\ldots\oplus\gen{\overline{g_m}^q}.\]
Then, assuming that we can efficiently compute in $\overline{G}_2$, \ref{it:A4} also holds: The orders of each $\overline{g_i}^q$ can be found in polynomial-time by repeated squaring as all of the orders are powers of two. 
 Then, as a consequence of \cref{lem:h_i<g_i}, the $\LHNPKE$ problem for $(\overline{G}_2,H_2)$ is reduced to the eDLP problem in $\overline{G}_2$ and the generator membership problem in $H_2$. 

Both this more general attack and the attack on the finite abelian case presented in \cref{main-theorem} rely in some capacity on the eDLP, the generator membership problem, and our assumptions \ref{it:A1}--\ref{it:A4}.
While the eDLP and the generator membership problem are relatively well-understood in the cases we need for our attacks, their complexity depends heavily on the instantiation, and we discuss these problems in \cref{sec:eDLP2gps} and \cref{sec:membership} respectively.
Also, while assumptions~\ref{it:A1}--\ref{it:A4} are natural assumptions to make when constructing groups that both admit efficient computation and satisfy the security requirements \ref{it:S1}--\ref{it:S4} of \cref{sec:security-assumptions}, 
there exist examples of groups where these assumptions are not satisfied or may be at odds with our security requirements.
In~\cref{RSAetc} and~\cref{A3justify} we discuss assumptions~\ref{it:A1}--\ref{it:A4} with a view towards constructing instantiations of Leonardi--Ruiz-Lopez encryption to which our classical attack does not apply, or at least is not polynomial-time.

\subsubsection{The eDLP in finite abelian $p$-groups}\label{sec:eDLP2gps}

Let $p$ be a prime number and let $G$ be a finite abelian $p$-group given as 
\[G=\gen{g_1}\oplus\ldots\oplus\gen{g_m},\]
where the orders of the summands are known.
Let $e\in\Z$ be such that $p^e$ is the exponent of $G$. Then, according to \cite[Cor.~1]{Sutherland}, the eDLP in $G$ can be solved using 
\begin{equation} \label{re} O\left(
\frac{\log(e+1)}{\log\log(e+2)}\log|G|+\frac{\log_p|G|}{m}p^{m/2}
\right)
\end{equation} 
group operations on a classical computer.
In particular, when $m$ is polynomial in $\log\log |G|$ and $p$ a fixed (small) prime, the eDLP has complexity polynomial in $\log|G|$.
When $p=2$, setting  $n= \log_p|G|$, we deduce the following from the performance result in  \cite[Table~1]{Sutherland}: 
\begin{itemize}
    \item When $m =1,2,4,8,$ the counts are dominated by the first term of \eqref{re}. This explains the initial cost decrease when $m$ increases for a fixed $n.$ 
    \item Using Shank's algorithm, there is a possibility of improving the factor $n/m$ by $\sqrt{n/m}$ though this is not relevant for applications: these normally require that $n/m$ is close to $1.$ 
\end{itemize}
To the best of our knowledge, there is no existing work on the eDLP that beats~\cite{Sutherland}. 

\subsubsection{The generator membership problem}\label{sec:membership}

Here we briefly discuss some known results on the generator membership problem in more generality than finite abelian 2-groups satisfying our assumptions:
Let $H$ be a finite group and let $Y$ be a (finite) subset of $H$. Below we collect some known results regarding membership testing for $\gen{Y}$, depending on how the group $H$ is given.

If $H$ is given as a subgroup of some finite $\GL(n,q)$ and $\mu$ denotes the largest prime dividing $|\gen{Y}|$ and not dividing $q$, then testing whether $x\in H$ belongs to $\gen{Y}$ can be achieved in time that is polynomial in $|Y|+n+q+\mu$; cf.\ \cite[Theorem~3.2(2)]{Luks}.
%\Chloe{This doesn't really fit here as we're talking about abelian groups. (And the result isn't close to giving us what we need, so I don't think we can use it)}\Mima{The paper is about solvable groups so it does, but I got your point that you cannot control the complexity well.}

%If $H$ satisfies \ref{it:A3}, then \cite[Lemma~3(a)]{BerVega14} ensures that, if $Y$ is of \emph{polynomial size}, then there is a polylog-time deterministic algorithm that tests membership in $\gen{Y}$.
%\Chloe{If $Y$ is not of polynomial-size the cryptosystem breaks down as $Y$ is smaller than the public key}\Mima{This is a comment we can make, but elsewere as this section does not deal with the cryptosystem}
\Chloe{I have commented out \cite{BerVega14} as this is now a section about why we have restricted tot eh case where we \emph{can} use that}

If $H$ is given as a permutation group and satisfies \ref{it:A3}, then membership testing in $Y$ is NC (polylogarithmic on a polynomial number of parallel processors); cf.\ \cite[Theorem~1.2(a)]{NC} and \cite[Theorem~1]{AbNC}.
\Chloe{\cite{NC} requires (A3). \cite{AbNC} doesn't seem to have a Theorem 1?}\Mima{found the problem: I meant to cite \url{https://ieeexplore.ieee.org/stamp/stamp.jsp?tp=&arnumber=4568072} (the titles are very similar)}
\Chloe{I can't open this reference! We could just remove it and leave it as \cite{NC} as this is now more exploratory?}\Mima{OK}

\Mima{Please feel free to add more. My understanding from the references that I checked is that it is not known if membership testing is polynomial time for permutation groups in general and I could not find anything specific complexity-wise for abelian permutation groups. Also, are there maybe more recent general references? I found one citing \cite{AbNC} but it did not look worth mentioning here.}

\subsubsection {On the security assumptions \ref{it:A2} and \ref{it:A4}} \label{RSAetc}
In \cref{sec:attack2gps} above we described how and under which assumptions the $\LHNPKE$ problem for $G$, $H$, and $K$ can be reduced to solving the eDLP in $G$ and the membership problem in $H$. 
The necessity for setting assumptions \ref{it:A2} and  \ref{it:A4} in particular, came from the fact that there are known examples of abelian groups, some of them already being used successfully in existing  cryptographic protocols, that do not have to obey them. 

One type of such groups are the known RSA groups. These are given in the form $(\mathbb{Z}/N\mathbb{Z})^{\times}$ where $N = pq$ is hard to factor. In this case, Alice would know the factorisation and hence the order of the group, but an adversary should not be able to compute it. 

Another type of such groups, where even the creators of the cryptosystem might not know the group's order, is that of ideal class groups of imaginary quadratic number fields. This is an interesting category of groups for cryptography since it allows one to work in a \emph{trustless setup}. In other words, we do not need a trusted third party to generate groups of secure order, in contrast to cryptosystems that employ RSA groups for example, where a trusted third party needs to generate a secure, i.e. hard to factor, modulus $N \in \mathbb{N}$ for the groups $(\mathbb{Z}/N\mathbb{Z})^{\times}$. Despite the fact that neither the structure nor the order of the ideal class group is known, the group operation is efficient and the elements of the group have a compact representation, via reduced binary quadratic forms. An excellent reference for trusted unknown-order groups is the paper by Dobson, Galbraith and Smith \cite{DGS};  
 they also take into account Sutherland's algorithm \cite{SuthClass} and they propose new security parameters for cryptosystems that employ ideal class groups. In the same paper the authors also discuss other groups of unknown order that can be used, namely genus-$3$ Jacobians of hyperelliptic curves, initially introduced by Brent \cite{Brent}. Even though these groups appear to have some computational advantages when compared to ideal class groups, %\Mima{again, I suggest to cut this}, 
 they offer less security than the genus-$1$ and genus-$2$ curves, as one can see for example in \cite{FWW}. 
 %\Vishnu{Is it common to acknowledge anonymous referee in the main text?}\Mima{I have never seen that done. Only in the acknowledgements, but we already thank them there.}

\subsubsection{Regarding Assumption \ref{it:A3}}\label{A3justify}

In this section, we discuss Assumption \ref{it:A3}. For the sake of simplicity and in view of \cref{sec:reduce_2}, we restrict to $2$-groups but everything can be said similarly for arbitrary finite abelian groups using \cref{result_abelian}. 

    The case in which Assumption \ref{it:A3} holds for the subgroup $\gen{g_1,\ldots,g_m}$ of $G$, i.e., $g_1,\dots,g_m$ satisfy $\langle g_1,\dots,g_m \rangle =\gen{g_1}\oplus\ldots\oplus\gen{g_m}$ (equivalently $g_1,\ldots, g_m$ are \emph{independent}), seems to be the key case for the following reasons. First of all, when sampling $g_1,\dots,g_m$ from $G$ uniformly, it is very likely that these are independent, at least if $m$ is small in comparison with the number of cyclic factors in $G$. 
    For example, if $G=C_2^\lambda$ is  a direct product of cyclic groups of order $2$, then the probability of sampling $g_1,\dots,g_m$ that are independent is
    $$ \frac{(2^\lambda-1)(2^\lambda-2)(2^\lambda-2^2)\cdots(2^\lambda-2^{m-1})}{2^{m\lambda}},$$ which is very close to $1$ if $m\ll \lambda$. 
    In addition, the case in which $g_1,\dots,g_m$ are independent can be seen as the generic case; 
    we expect that an attacker could use similar strategies as developed in \cref{sec:assumptions} and design an attack for the dependent case. Indeed, if $g_1,\dots,g_m$ are dependent, the attacker will obtain more information from the public key than in the independent case. \\

There is also another strategy for an attack if \ref{it:A3} does not hold for $G$ and the number of cyclic factors in $H$ is small. More precisely, assume that $|\Hom(H,\F_2)|$ is sub-exponential, or in other words, the number of maximal subgroups of $H$ is sub-exponential. In addition, it would be necessary to list these maximal subgroups $M$ and to be able to work in the quotients $H/M$ efficiently. If Eve can write any element $g \in G$ as a product in $g_1,\dots,g_m$  then she can recover the secret message $\beta$ from the ciphertext $(g,h)$ as follows. For every maximal subgroup $M$ of $H$ not containing $\tau$ (by assumption, there are only sub-exponentially many of these) convert $h$ into an element $\tilde h$ in $H/M$. Decrypt $(g,\tilde h) \in G\times \tilde H$ as if the noise was erased completely in $H/M$ (as explained in Section \ref{sec:noiseless}) and check on a number of self-encrypted messages if this provides a correct decryption function. Since $\ker(\psi)\neq H$, there will always be a maximal subgroup $M$ of $H$ that contains $h_1\dots,h_m$ but does not contain $\tau$, so eventually this procedure will indeed provide a decrpytion function. 

\subsection{Comparison with the quantum attack}\label{sec:LRL-cryptanalysis}

Throughout this section, assume that \ref{it:A3} holds for the abelian group $G$ together with $g_1,\ldots,g_m$. Even though we only did cryptanalysis for finite groups until now, we show in the following two subsections how one can perform quantum attacks when $G$ is either torsion-free or finite.
The torsion-free case was handled already in~\cite{Eprint} but we include it here for completeness; the finite case is new.
The mixed case can be considered as a combination of the two cases: first dealing with the free part of the group and then recovering $\beta$ as explained in the finite case.

\subsubsection{The torsion-free case}\label{subs:TorsionFree}

Assume as in \cite[Sec.\ 8.2]{Eprint}, that $G=\gen{g_1}\oplus\ldots\oplus\gen{g_m}$ is isomorphic to $\Z^{m}$, i.e.\ the orders $|\langle g_i \rangle|$ are all infinite.
We briefly recall the discussion from \cite[Sec.\ 8.2]{Eprint}. 
To this end, let $f:\Z^{m+1}\rightarrow G$ be defined by
\[
(a_1,\ldots,a_{m+1}) \longmapsto g^{a_{m+1}}\cdot g_1^{a_1}\cdots g_m^{a_m}. 
\]
 Then the kernel of $f$ is equal to $\gen{(r_1,\ldots,r_m,1)}$ and it can be determined, using Shor's algorithm~\cite{Shor_2}, in quantum polynomial time in $m$; cf.\ \cite{quantumtorsionfree}.
 %\cite{ImIv22,Wat01}. 
Once $(r_1,\ldots,r_m,1)$ is known, it is easy to recover $\beta$ from the encrypted message $(g,h)$ and the public information.

\subsubsection{The torsion case}\label{finiteabelianquantum}

Assume in this section that 
$G$ is finite and that \ref{it:A4} holds in addition to \ref{it:A1} and \ref{it:A3}. Let, moreover,  $f:\Z^{m+1}\rightarrow G$ be defined by
\[
(a_1,\ldots,a_{m+1}) \longmapsto g^{a_{m+1}}\cdot g_1^{a_1}\cdots g_m^{a_m}. 
\]
Then a set of generators of $\ker(f)$ can be determined in time polynomial in $\log|G|$ on a quantum computer \cite{Kitaev,Shor_2,Simon}; see also \cite{EHK04}. 
Note that $\ker(f)$ contains $|\langle g_1 \rangle|\Z\times \ldots\times |\langle g_m \rangle |\Z\times |\langle g \rangle|\Z$. 

Let now $(s_1,\ldots,s_{m+1})$ be one of the generators found. Then $g=g_1^{-s_1s_{m+1}}\cdots g_m^{-s_ms_{m+1}}$ and so it follows that 
\[
g_1^{r_1+s_1s_{m+1}}\cdots g_m^{r_m+s_ms_{m+1}}=1.
\]
This is the same as saying that, for each $1\leq i\leq m$, one has $$r_i\equiv -s_is_{m+1}\bmod |\langle g_i \rangle|. $$ 
Modding out $H$ by the subgroup generated by all $h_i^{|\langle g_i \rangle|}$  one can recover $\beta$ and thus solve the $\LHNPKE$ problem; cf.\ \cref{lem:h_i<g_i}.

\section{An attack strategy that reduces to the abelian case}\label{sectionConvert}
In this section, we assume $G$ and $H$ are finite.
A general strategy for an attacker to solve the $\LHNPKE$ problem could be to convert the $H$-part of the ciphertext as follows.   Assume that Eve has access to the ciphertext
\[(g,h) = (g_{w_1}\cdots  g_{w_\ell}, \varphi(g_{w_1})h_{w_1}\cdots\varphi(g_{w_{\ell}})h_{w_\ell}\tau^\beta) \in G \times H\]
encrypted by Bob using Alice's public key
\[\{(g_1,\varphi(g_1)h_1),\ldots,(g_m,\varphi(g_m)h_m),\tau\}.\]Eve can then choose a homomorphism $\vartheta\colon H \to \overline{H}$ for some group $\overline{H}$ such that $\vartheta(\tau)\neq 1$ and compute $\vartheta(h)$. Her new pair $(g,\overline{ h})=(g,\vartheta(h))$ is then of the form 
\[ (g_{w_1}\cdot \dotsc \cdot g_{w_{\ell}},\overline{\varphi}(g_{w_1})\overline{ h}_{w_1}\cdot \dotsc \cdot \overline{\varphi}(g_{w_{\ell}})\overline{ h}_{w_{\ell}} \overline{\tau}^\beta) \]  where, for all $i$, we set $\overline{h_i}=\vartheta(h_i)$ and write  $\overline{\varphi}=\vartheta\circ\varphi$ and $\overline{\tau}=\vartheta(\tau)$. Since $\overline{\tau}\neq 1$, this pair still contains the information on $\beta$ and if $\vartheta$ is chosen cleverly, it might be much simpler to deduce $\beta$ from $(g,h')$. Note that in general, we cannot define a suitable counterpart $\overline{\psi}$ of $\psi$ here, so the information that $\tau$ is not contained in the kernel of $\psi$ cannot be directly converted into a statement on $\overline{\tau}$ and has to be considered individually (if necessary).

\medskip

A special case of this strategy is to reduce, if possible, to a finite abelian group $\overline{H}$, i.e., the goal is to eventually apply the attack that we describe in~\cref{sec:limitations} even when $G$ and $H$ are nonabelian.  

Suppose that, given $\tau$ and $H$, Eve is able to find a finite abelian group $\overline{H}$ and an efficiently computable homomorphism $\vartheta: H \rightarrow \overline{H}$ such that $\vartheta(\tau) \neq 1$ and such that it can be checked that $\vartheta(\tau)$ is not contained in the subgroup generated by certain powers of $\vartheta(\varphi(g_i)h_i)$. This condition can be checked using the public key and we will specify below which powers are sufficient.

Eve follows the following three steps based on \cref{G_also_abelian}, \cref{result_abelian} and \cref{lem:h_i<g_i}.

\paragraph*{Step 1.} Let $\overline{G}=G/[G,G]$ be the abelianization of $G$ and $\bar \varphi \colon \overline{G} \to \overline{H}$ be the homomorphism obtained from $\vartheta\circ \varphi \colon G\to \overline{H}$ by reducing modulo $[G,G]$. For all $i$, define $\bar h_i=\vartheta(h_i)$, set $\bar \tau=\vartheta(\tau)$ and, for all $g$ in $G$, let $\bar g \in \overline{G}$ be the image of $g$ in $\overline{G}$. 

As in \cref{G_also_abelian}, Eve replaces the encrypted message by 
\[(\bar g, \vartheta(h))=(\bar g_{w_1}\cdots  \bar g_{w_\ell}, \bar\varphi(\bar g_{w_1})\bar h_{w_1}\cdots\bar\varphi(\bar g_{w_{\ell}})\bar h_{w_\ell}\bar \tau^\beta) \in \overline{G} \times \overline{H}\]
in order to work inside abelian groups $\overline{G}$ and $\overline{H}$.

\paragraph*{Step 2.} Provided \ref{it:A2} holds for the finite abelian groups $\overline{G}$ and $\overline{H}$, Eve proceeds in a similar way as in the proof of \cref{result_abelian}: First, she computes the largest positive odd factor of the group orders of $\overline{G}$ and $\overline{H}$ and finds a common multiple $q$ of their odd parts. 
Then she converts the tuple $(\bar g, \vartheta(h))  $ above into a tuple with entries inside the $2$-Sylow subgroups $\overline{G}_2$, $\overline{H}_2$ of $\overline{G}$ and $\overline{H}$ by taking both entries to their $q$-th powers: 
\[(\bar g^q, \vartheta(h)^q)=(\bar g_{w_1}^q\cdots  \bar g_{w_\ell}^q, \bar\varphi(\bar g_{w_1}^q)\bar h_{w_1}^q\cdots\bar\varphi(\bar g_{w_{\ell}}^q)\bar h_{w_\ell}^q\bar \tau^\beta) \in \overline{G}_2 \times \overline{H}_2.\]

\paragraph*{Step 3.} Finally, assuming \ref{it:A3} and \ref{it:A4} hold for $\overline{G}_2$, Eve proceeds in a  similar way as in the proof of \cref{lem:h_i<g_i}. 
For all $i$ define $\bar \ell_i=(\bar\varphi(\bar g_i)\bar h_i)^q=\vartheta(\varphi(g_i)h_i)^q$ and let $
\beta_i$ be the order of $\bar g_i^q$. 
Eve computes the elements $\bar \ell_i$ and the numbers $\beta_i$ using the public key. Then  $(\bar \ell_i)^{\beta_i}=(\bar h_i^q)^{\beta_i}$ can also be computed from the public key. Assuming the eDLP is solvable in $\overline{G}_2$ and the generator membership problem is solvable in $\overline{H}_2$, and if $\bar \tau$ is not contained in the subgroup $M$ generated by $(\bar \ell_1)^{\beta_1},\dots,(\bar \ell_m)^{\beta_m}$, the value of $\beta$ can be recovered as in the proof of \cref{lem:h_i<g_i}.
%{\color{lightgray} Old formulation: If $\bar \tau$ is not contained in the subgroup $M$ generated by $(\bar \ell_1)^{\beta_1},\dots,(\bar \ell_m)^{\beta_m}$,then working inside $\overline{H} /M$ can reveal the value of $\beta$ as in the proof of \cref{lem:h_i<g_i} under similar assumptions as stated there. }

%{\color{lightgray}NEW: Eve first reduces the problem to $2$-groups using \cref{result_abelian} and then uses \cref{lem:h_i<g_i} to solve eDLP for $2$-groups, since solving eDLP for $2$-groups is faster than solving eDLP for other groups. See Section 3.3.1}
%\Vishnu{Is this true? the eDLP depends both on $p$ and $m$, I do not know how m changes as we vary $p$ } \Annette{The new $m$ of the $2$-part of the  group is smaller or equal than the 'old' m, so I think what you are saying is true} \Mima{I am confused, why is this here and not in the next subsection?}\Vishnu{I moved it to the next step, but now I think we can skip this remark, so I put it in gray.}\Mima{I agree to removing it}

\begin{rmk}\label{rmk:attack-ab}
    Observe that the above attack applies to any finite group $H$ if $\tau \not\in [H,H]$ and the image of $\tau$ in $H/[H,H]$ is not contained in the subgroup generated by the images of certain powers of $\varphi(g_i)h_i$.
     In particular, it is advisable (but likely not sufficient) to sample $\tau$ from~$[H,H]$. 
\end{rmk}

\section{Normal forms and the solvable case}\label{sec:solvable}
 
In this section we consider the case in which $G$ and $H$ are finite and solvable and give evidence of why, given the efficiency constraints attached to the system, the groups should not be expected to provide secure postquantum cryptosystems. This is done in analogy to \cref{sec:LRL-cryptanalysis}, but with more iterations. Roughly speaking, since solvable groups have natural filtrations yielding abelian quotients, as we explain below, one could perform a quantum attack on solvable groups by iterating quantum attacks on abelian groups.

For the background on finite solvable groups we refer to the very friendly \cite[Ch.\ 3]{Isaacs}.

\begin{definition}
For a finite group $\Gamma$, the \emph{derived series} of $\Gamma$ is the series $(\Gamma^{(i)})_{i\geq 1}$ defined inductively by 
\[\Gamma^{(1)}=\Gamma \ \ \textup{ and } \ \  \Gamma^{(i+1)}=[\Gamma^{(i)},\Gamma^{(i)}].\]
If for some index $m$ the group $\Gamma^{(m)}$ is trivial, then $\Gamma$ is said to be \emph{solvable}.
\end{definition}

In a finite group $\Gamma$ each quotient $\Gamma^{(i)}/\Gamma^{(i+1)}$ is abelian and, if $\Gamma$ is solvable and $\Gamma^{(i)}/\Gamma^{(i+1)}$ is trivial, then $\Gamma^{(i)}=\Gamma^{(i+1)}=\Gamma^{(i+2)}=\ldots=\{1\}$.

Until the end of this section, assume that $H$ is solvable. Then, for cryptanalysis purposes and in analogy with \cref{G_also_abelian}, we assume without loss of generality that $G$ is also solvable. We let $s$ be the \emph{derived length} of $H$, i.e.\ $s$ is such that $H^{(s+1)}=1$ but $H^{(s)}\neq 1$. Since $\varphi(G^{(s+1)})$ is contained in $H^{(s+1)}=\{1\}$, in analogy to \cref{sec:Habelian}, we assume without loss of generality that $G^{(s+1)}=1$. We remark that, if $s=1$, then $G$ and $H$ are abelian.

\begin{rmk}[Efficient communication and computation]\label{rmk:efficiency}
Alice and Bob,  as part of their message exchange, need to be able to communicate elements and perform operations in the groups efficiently. An often favourable approach (also proposed in \cite[\S~4.2]{Baumslag}) is that of using \emph{normal forms} of elements with respect to a \emph{polycyclic presentation}, %as mentioned for instance in 
% \cite[\S~4.2]{Baumslag} or 
cf.\ \cite[Ch.~2]{EickHab}. For instance, when working with $G$ and $H$ abelian, \ref{it:A3} and \ref{it:A4} holding for $G$ is almost the same as saying that $G$ is given by a polycyclic presentation and the expression $g=g_1^{r_1}\cdots g_m^{r_m}$ is the normal form of $g$ with respect to this presentation. The word ``almost'' in the previous sentence is there to stress that the $r_i$'s are not uniquely identified by their class modulo $|\langle g_i\rangle|$, which in turn is what happens for normal forms (see below).
\end{rmk}

We briefly explain here what it means for an element $g$ of the solvable group $G$ to be communicated in a \emph{normal form} with respect to a polycylic presentation respecting the derived filtration. To do so, for each $i\in\{1,\ldots,s\}$:
\begin{enumerate}[label=$(\alph*)$]
\item let $m_i$ denote the minimum number of generators of $G^{(i)}/G^{(i+1)}$,
\item let $g_{i1},\ldots, g_{im_i}$ be elements of $G^{(i)}$ such that 
\[
G^{(i)}/G^{(i+1)}=\gen{g_{i1}G^{(i+1)}}\oplus\cdots\oplus\gen{g_{im_i}G^{(i+1)}},
\]
% so that, in accordance with the previous point,  $g_{11}=g_1, \ldots, g_m=g_{1m}$,
\item for each $j\in\{1,\ldots,m_i\}$ let $o_{ij}$ denote the order of $g_{ij}$ modulo $G^{(i+1)}$.
\end{enumerate}
Then any element $g\in G$ can be uniquely represented by a vector 
$$\delta=(\delta_{11},\ldots,\delta_{1m_1},\delta_{21},\ldots,\delta_{2m_2},\ldots,\delta_{s1},\ldots,\delta_{sm_s})$$
of integers $0\leq \delta_{ij}<o_{ij}$ of length $n=m_1+\ldots+m_s$  such that 
\[
g=\prod_{i=1}^s\prod_{j=1}^{m_i}g_{ij}^{\delta_{ij}},
\]
which is precisely the normal form of $g$ with respect to the chosen generators $\{g_{ij}\}$. 

Note that the data $(a)$-$(b)$-$(c)$ mentioned above should be public for Alice and Bob to be able to share the elements with each other. If the chosen generators fit into a \emph{polycyclic sequence}, then the group operation is performed through the \emph{collection process} \cite[\S~2.2]{EickHab}.
It should be mentioned that, though polycyclic presentations generally yield a good practical performance, it is, to the best of our knowledge, not clear whether multiplication in these presentations can always be performed in polynomial time \cite{LGSo98}. 
Note that the expression of $h'$ also depends on the vector $\delta$:
\[
h'=\prod_{i=1}^s\prod_{j=1}^{m_i}(\varphi(g_{ij})h_{ij})^{\delta_{ij}}.
\]
Shor's algorithm, applied on each level $G^{(i)}/G^{(i+1)}$, is polynomial in the log of the size of this quotient. Moreover, Shor's algorithm really does determine the vector $\delta$ of $g$ because of the condition $0\leq \delta_{ij}<o_{ij}$. In particular
Eve can recover the vector $\delta$ of $g$ and use it to compute 
\[
\tau^{\beta}=\left(\prod_{i=1}^s\prod_{j=1}^{m_i}(\varphi(g_{ij})h_{ij})^{\delta_{ij}}\right)^{-1}\cdot h.
\]
The complexity of this algorithm on each quotient $G^{(i)}/G^{(i+1)}$ is polynomial in $\log|G^{(i)}/G^{(i+1)}|$, which makes the overall complexity to be polynomial in $\log |G|$. 
There is therefore no separation in complexity between the algorithm being run by the adversary, and the one being run by the user. 
This implies that for the system to be secure, we would need to have $|G| > 2^{2^{\lambda}}$, 
making it hard to say if it would even be possible to represent elements of $G$ in a computer.

\begin{rmk}\label{classicalsolvable}
Given that the quotients of consecutive elements of the derived series are abelian, it is natural to ask whether the classical attack we designed for abelian groups could be generalised to an attack in the solvable context. It seems, however, that the nonabelian solvable case is substantially different from the abelian one. Among the limitations are: 
\begin{itemize}
    \item $\tau$ could for instance belong to $H^{(s)}$ while $g_1,\ldots, g_m$ will typically live in $H^{(1)}\setminus H^{(2)}$ (in the general setting we are indeed not necessarily working with normal forms); 
    \item without knowing $\varphi$, it is not at all clear at which depths in the derived filtration $\varphi(g_1),\ldots\varphi(g_m)$ are to be found in $H$ (as publicly given only with noise);
    \item dealing with quotients is more delicate as one always has to consider normal closures of subgroups.
\end{itemize}
\end{rmk}

\section{Future work}\label{future-work}

In future work, we plan to consider several types of finite groups $G$, $H$, and $K$ for instantiating Leonardi--Ruiz-Lopez encryption and explore whether we can either construct attacks for the corresponding cryptosystems or prove security results. 

A first candidate would be the group $C_2^\lambda$ for all groups $G,H,K$. As none of the classical attacks presented in this paper apply in this case, 
Leonardi--Ruiz-Lopez encryption might prove to be classically secure for this choice. 
Other abelian candidates are the RSA groups and ideal class groups  mentioned in Section \ref{RSAetc}.

As a first nonabelian example, we plan to work with certain $p$-groups and use strategies that allow us to circumnavigate the attacks presented in Section \ref{sec:solvable}. In particular, it would be beneficial to work with presentations of groups that are not based on normal forms and yet allow efficient computation. 
The advantage of working with nonabelian groups is that it may be possible to construct a post-quantum additive homomorphic cryptosystem.

%%%%%%%%%%%%%%%%%%%%%%%%%%%%%%%%%%%%%%%%


\begin{thebibliography}{10}

\bibitem{NC}
L.~Babai, E.~Luks, and A.~Seress.
\newblock Permutation groups in {NC}.
\newblock In {\em Proc. 19th ACM Symp. on Theory of Computing}, pages 409--420,
  1987.

\bibitem{Baumslag}
G.~Baumslag, N.~Fazio, A.~R. Nicolosi, V.~Shpilrain, and W.~E. Skeith, III.
\newblock Generalized learning problems and applications to non-commutative
  cryptography (extended abstract).
\newblock In {\em Provable security}, volume 6980 of {\em Lecture Notes in
  Comput. Sci.}, pages 324--339. Springer, Heidelberg, 2011.

\bibitem{Benaloh}
J.~Benaloh.
\newblock Dense probabilistic encryption.
\newblock In {\em Proceedings of the Workshop on Selected Areas of
  Cryptography}, page 120–128, 1994.

\bibitem{BerVega14}
J.~Bermejo-Vega and M.~Van~den Nest.
\newblock Classical simulations of {A}belian-group normalizer circuits with
  intermediate measurements.
\newblock {\em Quantum Inf. Comput.}, 14(3-4):181--216, 2014.

\bibitem{Brent}
R.~P. Brent.
\newblock {\em Public Key Cryptography with a Group of Unknown Order}.
\newblock Tech. Rep. Oxford University, 2000.

\bibitem{TFHE}
I.~Chillotti, N.~Gama, M.~Georgieva, and M.~Izabachène.
\newblock Tfhe: Fast fully homomorphic encryptionover the torus.
\newblock {\em Journal of Cryptology}, 33:34--91, 2020.

\bibitem{practicalFHE}
K.~Cong, R.~C. Moreno, M.~B. da~Gama, W.~Dai, I.~Iliashenko, K.~Laine, and
  M.~Rosenberg.
\newblock Labeled {PSI} from homomorphic encryption with reduced computation
  and communication.
\newblock In {\em ACM-CCS 2021}, Lecture Notes in Computer Science, page
  1135–1150, 2021.

\bibitem{DGS}
S.~Dobson, S.~D. Galbraith, and B.~Smith.
\newblock Trustless unknown-order groups.
\newblock Cryptology ePrint Archive, Paper 2020/196, 2020.
\newblock \url{https://eprint.iacr.org/2020/196}.

\bibitem{FHEW}
L.~Ducas and D.~Micciancio.
\newblock Fhew: Bootstrapping homomorphic encryption in less than a second.
\newblock In {\em Advances in Cryptology - {EUROCRYPT 2015}}, Lecture Notes in
  Computer Science, pages 617--640, 2015.

\bibitem{EickHab}
B.~Eick.
\newblock Algorithms for polycyclic groups, 2000.

\bibitem{EHK04}
M.~Ettinger, P.~Hoyer, and E.~Knill.
\newblock The quantum query complexity of the hidden subgroup problem is
  polynomial.
\newblock {\em Inform. Process. Lett.}, 91(1):43--48, 2004.

\bibitem{FWW}
X.~Fan, T.~Wollinger, and Y.~Wang.
\newblock Inversion-free arithmetic on genus 3 hyperelliptic curves and its
  implementations.
\newblock {\em International Conference on Information Technology: Coding and
  Computing (ITCC'05) - Volume II, Las Vegas, NV, USA, 2005, pp. 642-647 Vol.
  1, doi: 10.1109/ITCC.2005.179.}, pages 642--647.

\bibitem{FonGal}
C.~Fontaine and F.~Galand.
\newblock A survey of homomorphic encryption for nonspecialists.
\newblock {\em EURASIP Journal on Information Security 2007}, (013801), 2007.

\bibitem{GSW}
C.~Gentry, A.~Sahai, and B.~Waters.
\newblock Homomorphic encryption from learning with errors:
  Conceptually-simpler, asymptotically-faster, attribute-based.
\newblock In {\em Advances in Cryptology - Crypto 2013}, Lecture Notes in
  Computer Science, pages 75--92, 2013.

\bibitem{Isaacs}
I.~M. Isaacs.
\newblock {\em Finite group theory}, volume~92 of {\em Graduate Studies in
  Mathematics}.
\newblock American Mathematical Society, Providence, RI, 2008.

\bibitem{quantumtorsionfree}
S.~J.~L. Jr. and L.~H. Kauffman.
\newblock Quantum hidden subgroup algorithms: a mathematical perspective.
\newblock In {\em Quantum Computation and Information}, volume 305 of {\em AMS
  Contemporary Mathematics}, pages 139--202. 2002.

\bibitem{Kitaev}
A.~Y. Kitaev.
\newblock Quantum computations: algorithms and error correction.
\newblock {\em Uspekhi Mat. Nauk}, 52(6(318)):53--112, 1997.

\bibitem{Knuth}
D.~E. Knuth.
\newblock {\em The art of computer programming. Volume 1, Fundamental
  Algorithms}.
\newblock 1997.

\bibitem{LGSo98}
C.~R. Leedham-Green and L.~H. Soicher.
\newblock Symbolic collection using deep thought.
\newblock {\em LMS J. Comput. Math.}, 1:9--24, 1998.

\bibitem{Eprint}
C.~Leonardi and L.~Ruiz-Lopez.
\newblock Homomorphism learning problems and its applications to public-key
  cryptography.
\newblock Cryptology ePrint Archive, Report 2019/717, 2019.
\newblock \url{https://ia.cr/2019/717}.

\bibitem{Luks}
E.~M. Luks.
\newblock Computing in solvable matrix groups.
\newblock In {\em 33rd annual symposium on Foundations of computer science
  (FOCS). Proceedings, Pittsburgh, PA, USA, October 24--27, 1992}, pages
  111--120. Washington, DC: IEEE Computer Society Press, 1992.

\bibitem{LPR}
V.~Lyubashevsky, C.~Peikert, and O.~Regev.
\newblock On ideal lattices and learning with errors over rings.
\newblock {\em Journal of the ACM (JACM)}, 60(6):1--35, 2013.

\bibitem{AbNC}
P.~McKenzie and S.~A. Cook.
\newblock The parallel complexity of abelian permutation group problems.
\newblock {\em SIAM J. Comput.}, 16(5):880--909, 1987.

\bibitem{Paillier}
P.~Paillier.
\newblock Public-key cryptosystems based on composite degree residuosity
  classes.
\newblock In {\em International Conference on the Theory and Application of
  Cryptographic Techniques ({EUROCRYPT}'99)}, volume 1592, pages 223--238,
  1999.

\bibitem{Regev}
O.~Regev.
\newblock On lattices, learning with errors, random linear codes, and
  cryptography.
\newblock {\em Journal of the ACM (JACM)}, 56(6):1--40, 2009.

\bibitem{Shor_1}
P.~W. Shor.
\newblock Algorithms for quantum computation: discrete logarithms and
  factoring.
\newblock In {\em 35th {A}nnual {S}ymposium on {F}oundations of {C}omputer
  {S}cience ({S}anta {F}e, {NM}, 1994)}, pages 124--134. IEEE Comput. Soc.
  Press, Los Alamitos, CA, 1994.

\bibitem{Shor_2}
P.~W. Shor.
\newblock Polynomial-time algorithms for prime factorization and discrete
  logarithms on a quantum computer.
\newblock {\em SIAM J. Comput.}, 26(5):1484--1509, 1997.

\bibitem{Simon}
D.~R. Simon.
\newblock On the power of quantum computation.
\newblock {\em SIAM J. Comput.}, 26(5):1474--1483, 1997.

\bibitem{SuthClass}
A.~V. Sutherland.
\newblock Order computations in generic groups, 2007.
\newblock PhD Thesis, Massachusetts Institute of Technology.
  \url{https://math.mit.edu/~drew/sutherland-phd.pdf}.

\bibitem{Sutherland}
A.~V. Sutherland.
\newblock Structure computation and discrete logarithms in finite abelian
  {$p$}-groups.
\newblock {\em Math. Comp.}, 80(273):477--500, 2011.

\end{thebibliography}
\end{document}